\documentclass[12pt]{amsart}
\usepackage[latin1]{inputenc}
\usepackage{latexsym}
\usepackage{amsmath}
\usepackage{amssymb}
\usepackage{amsfonts}
\usepackage{bbm}                    
\usepackage{bbold}                  
\usepackage{textcomp}               
\usepackage{amstext}
\usepackage{enumerate}
\usepackage{units}
\usepackage{multicol}
 \newtheorem{thm}{Theorem}[section]
 \newtheorem{cor}[thm]{Corollary}
 \newtheorem{lem}[thm]{Lemma}
 \newtheorem{prop}[thm]{Proposition}
 \theoremstyle{definition}
 
 \theoremstyle{remark}

 \numberwithin{equation}{section}
\newcommand{\CC}{\mathbb{C}}

\newcommand{\RR}{\mathbb{R}}

\newcommand{\wt}[1]{\widetilde{#1}}
    
\newcommand{\SPn}[2]{\langle \,#1\,|\,#2\, \rangle}

\renewcommand{\le}{\leqslant} 
\renewcommand{\ge}{\geqslant}
\newcommand{\ple}{\prec}

\renewcommand{\imath}{i}
\newcommand{\dom}{\mathrm{dom}}

\renewcommand{\Im}{\mathrm{Im}\,}

\newcommand{\one}{\mathbbm{1}}
\renewcommand{\L}{\mathcal{L}}
\renewcommand{\d}{\Delta}
\newcommand{\C}{\mathcal{C}}

\newcommand{\s}{S}
\newcommand{\bbbone}{\one}
\renewcommand{\r}{R}
\newcommand{\F}{{\mathfrak F}}

\begin{document}

\author{M. K\"onenberg}
\address{Department of Mathematics and Statistics, Memorial University, St. John's, NL, Canada}
\email{mkonenberg@mun.ca}

\author{M. Merkli}
\email{merkli@mun.ca}

\author{H. Song}
\email{hs1858@mun.ca}

\title {Ergodicity of the Spin-Boson Model for arbitrary coupling strength}

\begin{abstract}
We prove that the spin-boson system is ergodic, for arbitrary strengths of the coupling between the spin and the boson bath, provided the spin tunneling matrix element is small enough.
\end{abstract}

\maketitle

\section{Introduction and main result}

The Hilbert space of pure states of the spin-boson system is ${\mathbb C}^2\otimes{\mathcal F}$, where
\begin{equation}
{\mathcal F} = \bigoplus_{n\geq 0} L^2_{\rm sym}({\mathbb R}^{3n},d^{3n}k)
\label{Fockspace}
\end{equation}
is the symmetric Fock space over the one-particle (momentum representation) space $L^2({\mathbb R}^3,d^3k)$. The spin-boson Hamiltonian is the self-adjoint operator (see \cite{Leggett}, equation (1.4))
\begin{equation}
H = -\textstyle\frac{1}{2}\Delta\sigma_x +\frac{1}{2}\varepsilon\sigma_z +H_R +\textstyle\frac12 q_0 \sigma_z\otimes\phi(h),
\label{1.3}
\end{equation}
where $\sigma_x$ and $\sigma_z$ are Pauli matrices,
\begin{equation}
\sigma_x = \left(
\begin{array}{cc}
0 & 1\\
1 & 0
\end{array}
\right), \qquad
\sigma_z = \left(
\begin{array}{cc}
1 & 0\\
0 & -1
\end{array}
\right),
\end{equation}
$\Delta\in{\mathbb R}$ and $\varepsilon\in{\mathbb R}$ are the `tunneling matrix element' and the `detuning parameter', respectively. We are using units in which $\hbar$ takes the value one. The free field Hamiltonian is given by
\begin{equation}
H_R = \int_{{\mathbb R}^3} |k| a^*(k) a(k) d^3k,
\end{equation}
where the creation and annihilation operators satisfy the canonical commutation relations $[a(k),a^*(l)]=\delta(k-l)$ (Dirac delta distribution). $q_0\in\mathbb R$ is the coupling constant, and $\phi(h)$ is the field operator, smeared out with a test function $h\in L^2({\mathbb R}^3,d^3k)$,
\begin{equation}
\phi(h) = \frac{1}{\sqrt{2}}\left( a^*(h) +a(h)\right) = \frac{1}{\sqrt{2}}\int_{{\mathbb R}^3} \left( h(k) a^*(k) +\overline{h}(k) a(k)\right)d^3k.
\end{equation}

 In \cite{Leggett}, Leggett et al. consider (among many other things) the average of $\sigma_z$  at time $t\ge 0$, when the spin starts (at $t=0$) in the state `up' and the environment starts in its thermal equilibrium. They call this quantity $P(t)$. For arbitrary $q_0$ fixed, they perform formal time-dependent perturbation theory in $\Delta$ (small) and establish the formula ((3.37) in \cite{Leggett})
\begin{equation}
P(t) = P(\infty)+[1-P(\infty)]\exp -t/\tau,
\label{legfor}
\end{equation}
where $P(\infty)=-\tanh(\beta\varepsilon/2)$ is the equilibrium value and 
\begin{equation}
\tau^{-1} = \Delta^2\int_0^\infty d t\cos(\varepsilon t) \cos\left[\frac{q_0^2}{\pi}\, Q_1(t)\right]e^{-\frac{q_0^2}{\pi}\, Q_2(t)}.
\label{t2}
\end{equation}
Here, 
\begin{eqnarray}
Q_1(t) &=& \int_0^\infty d\omega \frac{J(\omega)}{\omega^2}\sin(\omega t),\label{q1}\\
Q_2(t) &=& \int_0^\infty  d\omega \frac{J(\omega)(1-\cos(\omega t))}{\omega^2} \coth(\beta\omega/2), \label{q2}
\end{eqnarray}
where the {\em spectral density} of the reservoir is defined by 
\begin{equation}
J(\omega) = \textstyle\frac{\pi}{2}\omega^2\int_{S^2} |h(\omega,\Sigma)|^2 d\Sigma,\qquad \omega\geq 0,
\label{1.10}
\end{equation}
the integral being taken over the angular part in ${\mathbb R}^3$. The function $h$ is the form factor in \eqref{1.3}.\footnote{The spectral density is related to the Fourier transform of the reservoir correlation function $C(t)=\omega_{R,\beta}(e^{i tH_R}\varphi(h)e^{-i tH_R}\varphi(h))$ by $J(\omega)=\sqrt{\pi/2}\tanh(\beta\omega/2)[\widehat{C}(\omega)+\widehat{C}(-\omega)]$.} Of course, it is assumed in \cite{Leggett} that the integral in \eqref{t2} does not vanish, so that $\tau<\infty$ is a finite relaxation time. Assuming this as well in the present paper, we show in Corollary \ref{maincor} that the spin-boson system has the property of return to equilibrium, for arbitrary $q_0$ and small $\Delta$. Our main result, Theorem \ref{THM1}, implies the corollary. It describes completely the spectrum of the generator of dynamics, which is purely absolutely continuous covering $\mathbb R$, except for a simple eigenvalue at the origin.

\medskip

The spin-boson system is a $W^*$-dynamical system $({\mathcal H},{\frak M},\alpha)$, where ${\frak M}$ is a von Neumann algebra of observables acting on a Hilbert space $\mathcal H$ and where $\alpha^t$ is a group of $*$automorphisms of $\frak M$. The ``positive temperature Hilbert space''  is given by
\begin{equation}
{\mathcal H} = {\mathbb C}^2\otimes {\mathbb C}^2\otimes{\mathcal F}_\beta,
\label{2.1}
\end{equation}
where ${\mathcal F}_\beta$ is the Fock space 
\begin{equation}
{\mathcal F}_\beta = \bigoplus_{n\geq 0} L^2_{\rm sym}(({\mathbb R}\times S^2)^{n},(d u\times d\Sigma)^{n}).
\label{GluedFock}
\end{equation}
It differs from the `zero-temperature' Fock space \eqref{Fockspace} in that the single-particle space at positive temperature is the `glued' space $L^2({\mathbb R}\times S^2,d u\times d\Sigma)$ \cite{JP} ($d\Sigma$ is the uniform measure on $S^2$). ${\mathcal F}_\beta$ carries a representation of the CCR algebra. The represented Weyl operators are given by $W(f_\beta) = e^{i\phi(f_\beta)}$, where $\phi(f_\beta)=\frac{1}{\sqrt{2}}(a^*(f_\beta)+a(f_\beta))$. Here, $a^*(f_\beta)$ and $a(f_\beta)$ denote creation and annihilation operators on ${\mathcal F}_\beta$, smoothed out with the function
\begin{equation}
f_\beta(u,\Sigma) = \sqrt{\frac{u}{1-e^{-\beta u}}}\ |u|^{1/2} \left\{
\begin{array}{ll}
f(u,\Sigma), & u\geq 0\\
-\overline{f}(-u,\Sigma), & u<0
\end{array}
\right.
\label{2.3}
\end{equation}
belonging to $L^2({\mathbb R}\times S^2,d u\times d\Sigma)$. It is easy to see that the CCR are satisfied, namely,
\begin{equation}
W(f_\beta)W(g_\beta) = e^{-\frac{i}{2}{\rm Im}\SPn{f}{g}} W(f_\beta+g_\beta).
\label{ccr}
\end{equation}
The vacuum vector $\Omega$ represents the infinite-volume equilibrium state of the free Bose field, determined by the formula 
\begin{equation}
\label{thav}
\SPn{\Omega}{W(f_\beta)\Omega} = \exp\left\{ \textstyle-\frac14 \SPn{f}{\coth(\beta|k|/2)f}\right\},
\end{equation}
see also \cite{AW}. The CCR algebra is represented on \eqref{GluedFock} as $W(f)\mapsto W(f_\beta)$, for $f\in L^2({\mathbb R}^3)$ such that $ \SPn{f}{\coth(\beta|k|/2)f}<\infty$. We denote the von Neumann algebra of the represented Weyl operators by ${\mathcal W}_\beta$.

The doubled spin Hilbert space in \eqref{2.1} allows to represent any (pure or mixed) state of the two-level system by a vector, again by the GNS construction. This construction is as follows. Let $\rho$ be a density matrix on ${\mathbb C}^2$. When diagonalized it takes the form $\rho=\sum_i p_i |\varphi_i\rangle\langle\varphi_i|$, to which we associate the vector $\Psi_\rho = \sum_i \sqrt{p_i}\varphi_i\otimes\overline\varphi_i\in {\mathbb C}^2\otimes {\mathbb C}^2$ (complex conjugation in any fixed basis -- we will choose the eigenbasis of $H_\s$ given after \eqref{2.4} below). Then ${\rm Tr}(\rho A)=\SPn{\Psi_\rho}{(A\otimes\one_S)\Psi_\rho}$ for all $A\in{\mathcal B}({\mathbb C}^2)$ and where $\one_S$ is the identity in ${\mathbb C}^2$. This is the GNS representation of the state given by $\rho$ \cite{BR,MSB}. The von Neumann algebra of observables is 
\begin{equation}
\label{vna}
{\frak M}={\mathcal B}({\mathbb C}^2)\otimes\one_S\otimes{\mathcal W}_\beta \subset {\mathcal B}({\mathcal H}).
\end{equation}

The dynamics of the spin-boson system is given by
\begin{equation}
\label{dyn}
\alpha^t(A) =e^{i tL}Ae^{- itL}, \qquad A\in\frak M.
\end{equation}
It is generated by the self-adjoint Liouville operator acting on $\mathcal H$,
\begin{eqnarray}
L &=& L_0+\textstyle\frac{1}{2}q_0 V -\textstyle\frac{1}{2} q_0JVJ \label{2.4}\\
L_0 &=& L_\s+L_R, \label{2.4'}
\end{eqnarray}
where $L_\s=H_\s\otimes\bbbone_\s-\bbbone_\s\otimes H_\s$ with $H_\s=-\frac{1}{2}\Delta\sigma_x+\frac12\varepsilon\sigma_z$ the free two-level part and $L_\r=d\Gamma(u)$ is the second quantization of multiplication by the radial variable $u$, i.e. the free Bose part. The interaction operator in \eqref{2.4} is
\begin{equation}
V = \sigma_z\otimes\bbbone_\s\otimes\phi(h_\beta),
\end{equation}
where $h_\beta$ is the image of the form factor $h$ of \eqref{1.3} under the mapping \eqref{2.3}. The operator $J$ in \eqref{2.4} is the modular conjugation, which acts as
\begin{equation}
J( A\otimes\bbbone_\s\otimes W(f_\beta(u,\Sigma)) )J = \bbbone_\s\otimes\overline A\otimes W(\overline f_\beta(-u,\Sigma)),
\label{2.6}
\end{equation}
where $\overline A$ is the matrix obtained from $A$ by taking entrywise complex conjugation (matrices are represented in the eigenbasis of $H_\s$). Note that by \eqref{2.3}, we have $\overline f_\beta(-u,\Sigma) = -e^{-\beta u/2}f_\beta (u,\Sigma)$. By the Tomita-Takesaki theorem \cite{BR}, conjugation by $J$ maps the von Neumann algebra of observables \eqref{vna} into its commutant. In particular, $V$ and $JVJ$ commute. For more detail about this well-known setup we refer to \cite{JP,BFSrte,MSB} and references therein.

The vector representing the uncoupled ($q_0=0$) KMS state is
\begin{equation}
\Omega_{0,\rm KMS} =  \Omega_{S,\beta}\otimes \Omega,
\label{2.13}
\end{equation}
where $\Omega_{S,\beta}$ is the vector representative of the Gibbs density matrix $\rho_{S,\beta}\propto e^{-\beta H_S}$. For $\Delta=0$, we have
\begin{equation}
\Omega_{S,\beta,\Delta=0}=\frac{ e^{-\beta\varepsilon/4}\varphi_{++}+ e^{\beta\varepsilon/4}\varphi_{--}}{\sqrt{e^{-\beta\varepsilon/2}+e^{\beta\varepsilon/2}}}.
\label{2.13'}
\end{equation}
According to Kato's perturbation theory, the $(\alpha^t,\beta)$-KMS state on $\frak M$ is 
\begin{equation}
\label{equilstate}
\Omega_{\rm KMS} = \frac{e^{-\beta (L_0+\frac12 q_0V)/2}\Omega_{0,\rm KMS}}{\|e^{-\beta (L_0+\frac12 q_0V)/2}\Omega_{0,\rm KMS} \|}.
\end{equation}
One shows that $\Omega_{0,\rm KMS}$ is in the domain of $e^{-\beta (L_0+\frac12 q_0V)/2}$ for any $\d, q_0\in\mathbb R$ (see e.g. \cite{DJP,BFSrte,BR}).

Our analysis requires a regularity assumption on the form factor $h$. Let $\alpha\ge 0$. We say $h$ satisfies the {\bf Condition (${\rm A}_\alpha$)} if
\begin{equation}
\label{ffregularity}
(1+ |i\partial_u|^\alpha) (i h/u)_\beta\in L^2(\mathbb R\times S^2,du\times d\Sigma),
\end{equation}
where $(h/u)_\beta$ is obtained from $h/u$ via \eqref{2.3}.

\begin{thm}
\label{THM1}   The spectrum of $L$ is all of $\mathbb R$, for arbitrary $q_0,\Delta\in\RR$. For any $q_0\in\mathbb R$, $q_0\neq 0$, there is a constant $\d_0$ such that if $0<|\d|\le \d_0$, then we have the following.

{\rm \bf (a)} If {\rm (A$_\alpha$)}, \eqref{ffregularity}, holds for some $\alpha >3/2$, then $L$ has no eigenvalues except for a simple one at the origin, and $L\Omega_{\rm KMS}=0$.

{\rm  \bf (b)} If {\rm (A$_\alpha$)}, \eqref{ffregularity}, holds for some $\alpha> 2$, then the absolutely continuous spectrum of $L$ is all of $\mathbb R$ and the singular continuous spectrum of $L$ is empty.
\end{thm}

Admissible form factors satisfying (${\rm A}_\alpha$) with $\alpha>2$  are for instance 
$h(u)=u^{1/2}e^{-u^2}$, $h(u)=u^p\, e^{-u}$ or $h(u)=u^p\, e^{-u^2}$ with $p>3$. We mention that the `glueing' of the function $f$ into $f_\beta$ given in \eqref{2.3} can be done in various ways. In particular, the minus sign in the second line ($u<0$) can be changed into an arbitrary phase $e^{i\phi}$. This phase can be chosen to accommodate different form factors to satisfy (A$_\alpha$). A discussion of this has been given in \cite{FMRTE}.

\medskip

The spectral properties of $L$ given in Theorem \ref{THM1} imply readily that any initial state converges to the equilibrium state, see e.g. \cite{JP,BFSrte}.

\begin{cor}[Return to equilibrium]
\label{maincor}
Assume the conditions of Theorem \ref{THM1}, (b). For any normal state $\omega$ of \, $\frak M$ and any $A\in\frak M$, we have 
$$
\lim_{t\rightarrow\infty}\omega(\alpha^t(A)) = \SPn{\Omega_{\rm KMS}}{A\Omega_{\rm KMS}}.
$$
\end{cor}

{\em Remarks.\ } 1. Here, a state $\omega$ of $\frak M$ is called normal if it is represented by a vector $\psi\in\mathcal H$, $\omega(A) = \SPn{\psi}{A\psi}$ (see \cite{BR} for more detail).

2. The corollary shows that $\lim_{t\rightarrow\infty}P(t)=P(\infty)+O(\Delta)$, in accordance with Leggett et al.'s formula \eqref{legfor} (they only exhibit the lowest order term in $\d$). 


\bigskip

{\bf Outline of the strategy.\ } The Liouvillean $L$ \eqref{2.4} is unitarily equivalent to $\L$ \eqref{2.9}. We describe this transformation, inspired by \cite{Leggett}, in Section \ref{utsect}. The advantage of working with $\L$ is that the coupling constant $q_0$ appears in $\L$ in a uniformly bounded way as opposed to a linear function as in $L$ (see \eqref{2.9''}-\eqref{2.7}). This enables us to obtain results for all $q_0\in\RR$.

We analyze the eigenvalues of $\L$ in Section \ref{evsect}, using the conjugate operator method. We take for the conjugate operator $A_\nu$ a regularized version of the translation generator $A=d\Gamma(-i\partial_u)$. It is important to note that the ``spectral deformation'' technique cannot be applied here. This is so since the interaction is essentially given by (a spin operator times) a Weyl operator $W(f)=e^{i\phi(f)}$. When applying a spectral translation with parameter $\theta\in\CC$ to the interaction, the Weyl operator transforms into $W_\theta(f)=e^{i\theta A}W(f)e^{-i\theta A}= e^{\frac{i}{\sqrt 2}(a^*(f_\theta) +a(f_{\bar\theta}))}$. The operator $a^*(f_\theta) +a(f_{\bar\theta})$ is not self-adjoint for $\theta\not\in\RR$ and hence the interaction becomes huge and is not relatively bounded with respect to the number operator  $N$. It is not known how to
show analyticity of $(\theta,z)\mapsto e^{i\theta A}(\L-z)^{-1}e^{-i\theta A}\in \mathcal{B}(\mathcal{H})$ in this situation. The idea is then to assume, instead of analyticity in $\theta$, that only the first few real derivatives $\partial^\alpha_t|_{t=0} W_t(f)$ exist (we manage with $\alpha=1,\,2$). The $\alpha$-th derivative is the $\alpha$-fold commutator of $W$ with $A$, which is relatively bounded w.r.t. $N^{\alpha/2}$, becoming more singular with increasing $\alpha$. This presents a difficulty we have to overcome in our analysis, which is not present in previous works, to our knowledge. Indeed, the typically considered interaction is linear in field operators, so it is $N^{1/2}$-bounded. In this case, commutators with $A$ of all orders are as well $N^{1/2}$-bounded (see e.g. \cite{HuSp,BS,Merkli2000,DJ,FaMoSk}).

Using a positive commutator argument, we show in Theorem \ref{eigenvaluethm} that $\L$ has no eigenvalues except for a simple one at zero, with corresponding eigenvector the KMS state $\psi_{\rm KMS}$. Two important ingredients of the proof are: a regularity result on eigenvectors of $\L$ with the ensuing virial identity (Theorem \ref{regevectthm}) and a usually called a Fermi Golden Rule Condition on the effectiveness of the coupling. The latter is expressed here by the fact that Leggett et al.'s ``relaxation time'' $\tau$  is finite (which is also assumed in \cite{Leggett}). Regularity of eigenvectors based on positive commutator estimates has been shown before for Pauli-Fierz type models, see e.g. \cite{FM}. Our approach to showing instability of eigenvalues under perturbation via a positive commutator argument is inspired by \cite{BFSS,Merkli2000}.

We show in Section \ref{acsect} that the continuous spectrum of $\L$ is purely absolutely continuous. To do so, we control the boundary values  of the resolvent $(\L-z)^{-1}$, as ${\rm Im}z\rightarrow 0_+$ (see \eqref{c2}). More precisely, we show that $\SPn{\varphi}{(\L-z)^{-1}\psi}$ is bounded as ${\rm Im}z\rightarrow 0_+$, for any $\varphi,\psi$ in a dense set, in the following way. Using the Feshbach map, we relate the resolvent to a ``reduced resolvent'' $(\bar \L-z)^{-1}$ and a ``Feshbach part'' ${\frak F}(z)^{-1}$, see \eqref{c3}. The reduced resolvent acts on the reduced Hilbert space ${\rm Ran}\bar P_\Omega$, while ${\frak F}(z)^{-1}$ is an operator on ${\rm Ran}P_\Omega$ (of dimension four). The control (boundedness) of the boundary values of $(\L-z)^{-1}$ is implied by that of  $(\bar\L-z)^{-1}$ and ${\frak F}(z)^{-1}$, shown in Theorems \ref{thmc0} and \ref{thmc2}, respectively. To prove Theorem \ref{thmc0}, we analyze the reduced resolvent based on a suitable approximation $(\bar\L(\eta)-z)^{-1}$, $\eta>0$, with $\bar\L(0)=\bar \L$. Regularizations of this type have are often used in ``Mourre theory''. They have been introduced in \cite{Mourre} and have also been used in \cite{ABG,BS,DJ,FaMoSk,HuSp}. We show that $\partial_z (\bar\L(\eta)-z)^{-1}$ is H\"older continuous in $\eta>0$, weakly on a dense set of vectors and uniformly in $\Im z>0$. This implies that $(\bar\L-z)^{-1}$ has a bounded extension to $\Im z=0_+$. In order to prove Theorem \ref{thmc2}, namely boundedness of the boundary values of ${\frak F}(z)^{-1}$, we first use the proven regularity of $(\bar\L-z)^{-1}$ to derive the existence of boundary values of ${\frak F}(z)$, as $\Im z\rightarrow 0_+$. We then show the invertibility of ${\frak F}(x)$, $x\in\RR\backslash\{0\}$, by using the fact that the only eigenvalue of $\L$ is zero and is simple.

\section{Unitary transformation} 
\label{utsect}

By a suitable unitary transformation, the Hamiltonian \eqref{1.3} with $\Delta=0$ can be diagonalized explicitly, see (3.28) of \cite{Leggett}. We modify this idea for application to the Liouville operator \eqref{2.4}. The unitarily transformed Liouville operator is
\begin{eqnarray}
\qquad \L &=& U L U^* = {\mathcal L}_0 +\d I\label{2.9}\\
\L_0 &=& \L_S+\L_R =  \textstyle\frac\varepsilon2(\sigma_z\otimes\bbbone_\s-\bbbone_\s\otimes\sigma_z) +L_\r\label{2.9'}\\
I&=& -\textstyle\frac12({\mathcal V}-J{\mathcal V}J)\label{2.9''}\\
{\mathcal V} &=& \sigma_+\otimes\bbbone_\s\otimes W(2f_\beta) +\sigma_-\otimes\bbbone_\s\otimes W(-2f_\beta).
\label{2.9'''}
\end{eqnarray}
The raising and lowering operators are given by
$$
\sigma_+=\left(
\begin{array}{cc}
0 & 1\\
0 & 0
\end{array}
\right)\qquad \sigma_- = \left(
\begin{array}{cc}
0 & 0\\
1 & 0
\end{array}
\right)
$$
and 
\begin{equation}
f_\beta = (-\textstyle\frac{i}{2}q_0 h/u)_\beta,
\label{2.7}
\end{equation} 
where $h$ and $q_0$ are the form factor and coupling constant given in the interaction in \eqref{1.3}, with $f\mapsto f_\beta$ given in \eqref{2.3}. Note that ${\mathcal V}$ is self-adjoint and bounded and satisfies ${\mathcal V}^2=\bbbone$. Since $\|\mathcal V\|=1$, we have $\|I\|\le 1$. Define the unitary operator
\begin{equation}
U = \exp i\big[ \sigma_z\otimes\bbbone_\s\otimes\phi(f_\beta) - J\{\sigma_z\otimes\bbbone_\s\otimes\phi(f_\beta)\}J\big],
\label{2.8}
\end{equation} 
where the action of $J$ is given in \eqref{2.6}.  Note that $U$ depends on the coupling parameter $q_0$. For the uncoupled system $q_0=0$, we have $U=\one$. 
The KMS vector associated to $\L_0$ is
\begin{equation}
\label{psiokms}
\psi_{0,\rm KMS} = \psi_{S,\beta}\otimes \Omega =  U \Omega_{0,\rm KMS} \quad \mbox{where \ \ $\psi_{S,\beta}:=\Omega_{S,\beta,\Delta=0}$}.
\end{equation}
By Araki's perturbation theory, the KMS vector associated to $\L$ is 
\begin{equation}
\psi_{\rm KMS} = \frac{e^{-\beta (\L_0-\frac12 \Delta{\mathcal V})/2}\psi_{0,\rm KMS}}{\|e^{-\beta (\L_0-\frac12\Delta {\mathcal V})/2}\psi_{0,\rm KMS} \|} = U\Omega_{\rm KMS}.
\label{2.12}
\end{equation}

\begin{thm}
\label{thm1} 
The spectrum of $L$ is all of $\mathbb R$, for arbitrary $q_0,\Delta\in\RR$. For any $q_0\in\mathbb R$, $q_0\neq 0$, there is a constant $\d_0$ such that if $0<|\d|\le \d_0$, then we have the following.

{\rm \bf (a)} If {\rm (A$_\alpha$)}, \eqref{ffregularity}, holds for some $\alpha >3/2$, then $\L$ has no eigenvalues except for a simple one at the origin, and $\L\psi_{\rm KMS}=0$.

{\rm  \bf (b)} If {\rm (A$_\alpha$)}, \eqref{ffregularity}, holds for some $\alpha>2$, then the absolutely continuous spectrum of $\L$ is all of $\mathbb R$ and the singular continuous spectrum of $\L$ is empty.
\end{thm}

The proof of Theorem \ref{THM1} follows immediately from this result and relation \eqref{2.9}.

\section{Proofs: Eigenvalues of $\L$}
\label{evsect}

\subsection{Conjugate operator}

We will assume throughout this section that \eqref{ffregularity} is satisfied for $\alpha>3/2$. Let $0<\nu\le 1$, $0<\epsilon<\alpha-3/2$  and set 
\begin{equation*}
w_\nu(u) = \int_0^u \frac{d s}{(\nu|s|+1)^{1+\epsilon}},\qquad u\in\mathbb R.
\end{equation*}
The derivative $w'_\nu(u)=(\nu|u|+1)^{-1-\epsilon}$ is strictly positive and converges to the constant function one as $\nu$ tends to zero. We abbreviate
$$
w_\nu=w_\nu(-i\partial_u),\quad w'_\nu=w'_\nu(-i\partial_u).
$$
The $\epsilon$ is arbitrary but fixed, determined by the regularity of the form factor, see \eqref{ffregularity}. We define the self-adjoint operators
$$
A_\nu=d\Gamma(w_\nu),\quad \C_\nu = d\Gamma(w'_\nu).
$$
The domains of both operators contain that of $N= d\Gamma({\mathbf 1})$, and the inequalities  $0<\C_\nu\le N$ and $\pm A_\nu \le N/\epsilon\nu$ hold in the sense of quadratic forms on $\dom(N)$. Moreover, $A_\nu$, $\C_\nu$ and $N$ commute on $\dom(N)$ and as a quadratic form on $\dom(N)\cap\dom(\L_R)$, we have
\begin{equation}
i[A_\nu,\L_R] = \C_\nu.
\label{m1}
\end{equation}

\begin{lem}\label{LemOperators}
1. For $g\in \dom((w_\nu')^{-1/2})$ and $\psi\in \dom(\C_\nu^{1/2})$,
\begin{align}
\label{RelBounds}
\|a(g)\psi\|^2&\le \|(w'_\nu  )^{-1/2}g\|^2\,\|\C_\nu^{1/2}\psi\|^2,\\ \nonumber
\|a^*(g)\psi\|^2&\le \|(w'_\nu)^{-1/2}g\|^2\,\|\C_\nu^{1/2}\psi\|^2+\|g\|^2\,\|\psi\|^2.
\end{align}

2. For $\psi\in\dom(\C_\nu^{1/2})$,
\begin{equation}
|\SPn{\psi}{i[ A_\nu, I] \psi}|\le c_1 \|\C_\nu^{1/2}\psi\|\,\|\psi\|+ c_2\|\psi\|^2,
\label{m6}
\end{equation}
where $c_1=4\sqrt{2} \|(1+|i\partial_u|^{3/2+\epsilon})f_\beta\|$ and  $c_2=c_1 (1+\|f_\beta\|)/\sqrt{2}$.
\end{lem}

The inequality \eqref{m6} implies that for all $\alpha>0$,
$i[A_\nu,I] \ge -\alpha c_1\C_\nu -(\textstyle\frac{c_1}{4\alpha} +c_2)$, 
as a quadratic form on $\dom(N)$. In combination with \eqref{m1} we obtain that for any $\alpha>0$,
\begin{equation}
i[A_\nu,\L] \ge (1-\alpha |\d| c_1)\C_\nu-|\d|(\textstyle\frac{c_1}{4\alpha} +c_2),
\label{m7}
\end{equation}
as a quadratic form on $\dom(N)\cap\dom(\L_R)$. 

\begin{proof}[Proof of Lemma \ref{LemOperators}]
1. The relative bounds \eqref{RelBounds} are most easily obtained by applying the Fourier transform (so that functions of $-i\partial_u$ become multiplication operators in the Fourier variable). Their derivation is standard, see e.g. \cite{BFS2}. 

2. Let $D$ be a self-adjoint operator on $L^2(\mathbb R\times S^2)$ and let $f\in\dom(D)$. As a quadratic form on $\dom(d\Gamma(D))$, we have
\begin{equation}
[d\Gamma(D),W(f)] = W(f)\big(\phi(iDf)+\textstyle\frac12\SPn{f}{Df}\big),
\label{m4}
\end{equation}
where $\phi$ is the field operator. This relation is readily obtained by taking the derivative $-i\partial_t|_{t=0}$ of $e^{itd\Gamma(D)} W(f) e^{-itd\Gamma(D)} = W(e^{itD}f)$. According to \eqref{2.9''}, \eqref{2.9'''} the interaction $I$ consists of four similar terms. We treat the part $-\frac12\sigma_+\otimes\one_S\otimes W(2f_\beta)$, the others are dealt with in the same way. Taking into account \eqref{m4} with $D=w_\nu$, we obtain for $\psi\in\dom(\C_\nu^{1/2})$
\begin{equation}
\textstyle\frac12 |\SPn{\psi}{\sigma_+\otimes\one_S\otimes [A_\nu, W(2f_\beta)]\psi}| \le  \|\psi\|\,\big(  \|\phi(iw_\nu f_\beta)\psi\| + \|\,|w_\nu|^{1/2} f_\beta\|^2\,\|\psi\|\big).
\label{m5} 
\end{equation}
Using \eqref{RelBounds} gives 
$$
\|\phi(iw_\nu f_\beta)\psi\|\le \sqrt{2} \|w_\nu(w'_\nu)^{-1/2}f_\beta\| \,\|\C_\nu^{1/2}\psi\| + \|w_\nu f_\beta\|\|\psi\|.
$$ 
Next, $|w_\nu(u)(w'_\nu(u))^{-1/2}|\le |u|^{3/2+\epsilon}$ (as $\nu\le 1$) and $|w_\nu(u)|\le |u|$. So both norms  $\|w_\nu(w'_\nu)^{-1/2}f_\beta\|$ and $\|w_\nu f_\beta\|$ are bounded above by $\|(1+|i \partial_u|^{3/2+\epsilon})f_\beta\|$. Furthermore, 
$$
\|\,|w_\nu|^{1/2}f_\beta\|^2\le \|f_\beta\|\,\|\, |i\partial_u|f_\beta\| \le \|f_\beta\|\, \|(1+|i\partial_u|)^{3/2+\epsilon}f_\beta\|.
$$
This shows \eqref{m6} and concludes the proof of Lemma \ref{LemOperators}.
\end{proof}

%
%
\subsection{Regularity of eigenvectors of $\L$.}

Let $0\le \chi\le 1$ be a smooth function which
satisfies $\chi(x)=1$ for $|x|\le 1/2$ and $
\chi(x)=0$ for $|x|\ge 1$. We set
\begin{equation*}
\chi_{\mu} = \chi((N+1)/\mu),\qquad \chi^{(n)}_{\mu}=\chi^{(n)}((N+1)/\mu), 
\end{equation*}
where $\chi^{(n)}$ denotes the $n$-th derivative of $\chi$, and where $\mu\ge 1$.

\begin{lem}
\label{Lem:CommutatorN}
1. The $k$-fold commutator ($k\ge 1$) of $N$ with $\L$, $ad_N^{(k)}(\L)=[N,[N,\ldots [N,\L]\ldots]]$, is relatively $N^{k/2}$ bounded, and 
\begin{equation}
\|ad_N^{(k)}(\L) (N+1)^{-k/2}\|\le |\d| c(k) (1+\|f_\beta\|)^{2k},
\label{m8}
\end{equation}
where $c(k)$ is independent of $f_\beta$.

2. On $\dom(\L_R)$,
\begin{align}\label{eq:CommutatorChi2}
[\chi_\mu,\L] = \mu^{-1}\chi'_\mu \, [N,\L]
  -\textstyle\frac12 \mu^{-2} \chi''_\mu\, [N,[N,\L]] + \d \mu^{-3/2} R_\mu,
\end{align}
with $\sup_{\mu\ge 1}\|R_\mu\|<\infty$.
\end{lem}

\begin{proof}
1. The operators $N$ and $\L_0$ commute, only the interaction contributes to the commutator. Using repeatedly \eqref{m4} with $D=\one$, together with the form equality $[N,a^*(f)]=a^*(f)$ (and its adjoint), one readily sees that $ad_N^{(k)}(\L)$ is a sum of four terms, each of the form $S\otimes W(f_\beta)T_k$, where $S$ is one of $\sigma_\pm\otimes\one_S$ or $\one_S\otimes\sigma_\pm$, and $T_k$ is a polynomial in $a^*(f_\beta)$, $a(f_\beta)$ (of maximal joint degree $k$). The relative bound follows.

2. By means of the Helffer-Sj\"ostrand formula \cite{Davies},
\begin{align}\label{HelfferSjoestrand}
\chi_\mu^{(n)}= (-1)^n n!  \int_{\CC} 
\partial_{\bar{z}}\widetilde{\chi}(z)(\textstyle\frac{N+1}{\mu}-z)^{-1-n}dz 
\end{align}
for $n=0,1,2\ldots$ We have, strongly on $\dom(\L)=\dom(\L_R)$,
\begin{align*}
[\chi_\mu,\L]=\d \int_{\CC} 
\partial_{\bar{z}}\widetilde{\chi}(z)\big[(\textstyle\frac{N+1}{\mu}-z)^{-1},I \big] dz.
\end{align*}
Using the relations \eqref{HelfferSjoestrand} and $[A^{-1},B]=A^{-1}[B,A]A^{-1}$, we arrive at
\begin{equation}
[\chi_\mu,\L] = \d\mu^{-1} \chi'_\mu\, [N,I] +\textstyle\frac12 \d\mu^{-2} \chi_\mu'' \, [N,[N,I]] +\d \mu^{-3/2} R_\mu,
\end{equation}
where 
\begin{equation}
R_\mu = \mu^{-3/2} \int_{\CC} 
\partial_{\bar{z}}\widetilde{\chi}(z) (\textstyle\frac{N+1}{\mu}-z)^{-3} ad_N^{(3)}(I) (\textstyle\frac{N+1}{\mu}-z)^{-1} dz.
\end{equation}
Invoking the relative bound \eqref{m8} and that $|{\rm Re} z|$, $|{\rm Im} z|\le 2$ (since $z$ is in the support of the almost-analytic extension $\partial_{\bar z}\widetilde\chi(z)$), we get
$$
\|  (\textstyle\frac{N+1}{\mu}-z)^{-3} ad_N^{(3)}(I) (\textstyle\frac{N+1}{\mu}-z)^{-1}\| \le C \mu^{3/2} |{\rm Im} z|^{-4},
$$
with a constant $C$ independent of $\mu$ and of $z$. However, $|\partial_{\bar z}\widetilde\chi(z)|\le C'|{\rm Im} z|^4$ for some constant $C'$ and so $\sup_{\mu\ge 1} \|R_\mu\|<\infty$.
\end{proof}

\begin{thm}[Regularity of eigenvectors]
\label{regevectthm}
Let $\psi$ be a normalized eigenvector of $\L$. Then $\psi\in\dom(N^{1/2})$ and for every $0<\xi<1$,
\begin{equation}
\label{ParticleNumberBound}
\|N^{1/2}\psi\|^2\le \frac{\xi^{-1}\d^2 c_1^2/4+|\d| c_2}{1-\xi}.
\end{equation}
Let $A\equiv A_{\nu=0}\equiv d\Gamma(-i\partial_u)$. The commutator $i[A,\L]$ is well defined as a quadratic form on $\dom(N^{1/2})$ and we have the virial identity
\begin{equation}
\label{Thm:Virial}
\SPn{\psi}{i [A,\L] \psi} =0.
\end{equation}
\end{thm}
{\em Remarks.\ } 1. In \eqref{Thm:Virial}, the commutator $i[A,\L]$ is understood as the closure of the sesquilinear form, defined on $\dom(A)\cap\dom(\L)$ by $i\SPn{A\varphi}{\L\psi} -i\SPn{\L\varphi}{A\psi}$. The self-adjoint operator associated to the closure of this form is $i[A,\L]=N+ \d i[A,I]$, where
\begin{eqnarray}
\label{m13}
i[A,I]  &=& -i\sigma_+\otimes\one_S\otimes W(2f_\beta)\left(\phi(f'_\beta) -i\SPn{f_\beta}{f'_\beta}\right)\\
&& +i\ \one_S\otimes \sigma_+\otimes J_RW(2f_\beta)\left(\phi(f'_\beta) -i\SPn{f_\beta}{f'_\beta}\right)J_R\nonumber\\
&& \mbox{$+$ adjoint.}\nonumber
\end{eqnarray}
The virial relation \eqref{Thm:Virial} needs a proof since $\psi$ is generally not in $\dom(A)$.

2. This result does not require $\d$ to be small.

\begin{proof}[Proof of Theorem \ref{regevectthm}]
Since the operator $A_{\nu,\mu}:= \chi_\mu A_\nu \chi_\mu$ is self-adjoint and bounded and $\psi\in\dom(\L)$, we have the virial identity
\begin{equation}
\label{Viral}
0=\SPn{\psi}{i [A_{\nu,\mu},\L]\psi}=t_1+t_2,
\end{equation}
with $t_1=\SPn{\chi_\mu\psi}{ \imath [A_\nu,\L] \chi_\mu \psi}$ and $t_2=2{\rm Re} \, \imath\SPn{\psi}{[\chi_\mu, \L]A_\nu \chi_\mu \psi}$. Choosing $\alpha = \xi(|\d| c_1)^{-1}$ in \eqref{m7} gives the lower bound
\begin{equation}
t_1\ge  (1-\xi) \, \SPn{\chi_\mu\psi}{ \C_\nu \chi_\mu \psi}-\textstyle\frac{\d^2 c_1^2}{4\xi} -|\d| c_2.
\label{m11}
\end{equation}
The expansion \eqref{eq:CommutatorChi2}, together with the bound $\|A_\nu(N+1)^{-1}\|\le 1/\epsilon\nu$ implies that 
\begin{align*}
|t_2|\le& \, 2 |\d|\ \mu^{-1}\ |\SPn{\psi}{\chi'_\mu [N, I] A_\nu \chi_\mu\psi}| \nonumber\\
&+ 2|\d| (\epsilon\nu)^{-1}\, c(2)\ (1+\|f_\beta\|)^4 \  \|\chi''_\mu\psi\|
+ 2|\d| (\epsilon\nu\mu^{1/2})^{-1}\|R_\mu\|.
\end{align*}
Recall that $c(k)$ is defined in Lemma \ref{Lem:CommutatorN}. Proceeding as in the proof of that lemma, point 1., one shows that for all $\varphi\in\dom(N^{1/2})$, $\|[N,I]\varphi\|\le 8\|(\phi(if_\beta)+\|f_\beta\|^2)\varphi\|$. Combining this estimate with \eqref{RelBounds} gives
\begin{align*}
& |\SPn{\psi}{\chi'_\mu [N, I] A_\nu \chi_\mu\psi}|  \le \|\chi'_\mu \psi\|\, \|[N,I]A_\nu \chi_\mu \psi\|\\
& \le 8\sqrt2 (\|f'_\beta\|+\|f_\beta\|  )\ \|\chi'_\mu\psi\| \ \|\C^{1/2}_\nu A_\nu\chi_\mu \psi\| +8\mu (\epsilon\nu)^{-1}\|\chi'_\mu\psi\|\,\|f_\beta\| (1+\|f_\beta\|).
\end{align*}
The $\C_\nu$, $\chi_\mu$ and $A_\nu$ commute and $\|\C_\nu^{1/2}\chi_\mu A_\mu\psi\|\le \mu(\epsilon\nu)^{-1}\|\C_\nu^{1/2}\chi_\mu\psi\|$. We make use of 
$$
\|\chi'_\mu\psi\|\,\|\C_\nu^{1/2}\chi_\mu\psi\| \le \widetilde\alpha\SPn{\chi_\mu\psi}{\C_\nu\chi_\mu\psi} +(4\widetilde\alpha)^{-1} \|\chi'_\mu\psi\|^2,
$$ 
with $\widetilde\alpha = \kappa\epsilon\nu [16\sqrt2 |\d|^{-1} (\|f_\beta\|+\|f'_\beta\|)]^{-1}$, for an arbitrary $\kappa>0$. This gives
\begin{align}
|t_2| & \le  \kappa\SPn{\chi_\mu\psi}{\C_\nu\chi_\mu\psi} + C|\d|(\epsilon\nu)^{-1}\|\chi''_\mu\psi\|\nonumber \\
&\ +C|\d|(\epsilon\nu)^{-1}\|\chi'_\mu\psi\|(|\d|(\epsilon\nu\kappa)^{-1}\|\chi'_\mu\psi\|+1)
+C|\d|(\epsilon\nu\mu^{1/2})^{-1},
\label{m10} 
\end{align}
where $C$ is a constant independent of $\d,\mu,\nu,\kappa$. The spectral support of the operators $\chi'_\mu$, $\chi''_\mu$ is contained in $\mu/2\leq N+1\leq \mu$. Thus we have $\lim_{\mu\rightarrow \infty}\|\chi'_\mu\psi\|=0=\lim_{\mu\rightarrow \infty}\|\chi''_\mu\psi\|$. It follows from \eqref{m10} that there exists a $\mu_0(\nu,\kappa)$ such that for $\mu\ge \mu_0$, we have 
\begin{equation}
|t_2|\le \kappa\SPn{\chi_\mu\psi}{\C_\nu\chi_\mu\psi} +\kappa.
\label{Eq4}
\end{equation}
Combining \eqref{Viral}, \eqref{m11} and \eqref{Eq4} gives
\begin{equation*}
\SPn{\chi_\mu\psi}{ \C_\nu \chi_\mu \psi}\le  a\equiv \frac{\xi^{-1}\d^2 c_1^2/4+|\d| c_2+\kappa}{1-\xi-\kappa}
\end{equation*}
whenever $\mu\geq\mu_0$. Note that $\C_\nu$ is self-adjoint and positive. Since $a$ does not depend on $\mu$, one easily shows, by taking $\mu\rightarrow\infty$, that $\psi\in\dom(\C_\nu^{1/2})$ and $\|\C^{1/2}_\nu\psi\|\le \sqrt a$. Next we take $\nu\downarrow 0$. According to the decomposition of Fock space into a direct sum of $n$-particle sectors, we have
$$
\SPn{\psi}{\C_\nu\psi} = \sum_{n\geq 1}\sum_{j=1}^n\SPn{\psi_n}{ [w'_\nu]_j\ \psi_n},
$$
where $[w'_\nu]_j$ is the operator $(\nu|i\partial_{u_j}|+1)^{-1-\epsilon}$, acting on the $j$-th radial variable, $u_j$, of $n$-particle sector $\psi_n(u_1,\Sigma_1,\ldots,u_n,\Sigma_n)$. Since  $[w'_\nu]_j\uparrow 1$ as $\nu\downarrow 0$ we invoke the monotone convergence theorem to conclude that $\lim_{\nu\downarrow0} \SPn{\psi}{\C_\nu\psi} =\SPn{\psi}{N\psi}\le a$. Upon taking $\kappa\rightarrow 0$ we obtain the bound \eqref{ParticleNumberBound}.

Next we prove \eqref{Thm:Virial}. We know from \eqref{Eq4} that $|t_2|\leq \kappa(a+1)$, provided $\mu\ge \mu_0$. Taking first $\mu\rightarrow\infty$ and then $\kappa\rightarrow 0$ in \eqref{Viral} gives 
\begin{equation}
\label{m12}
\lim_{\mu\rightarrow\infty} \SPn{\chi_\mu\psi}{i[A_\nu,\L]\chi_\mu\psi} =0.
\end{equation}
We have $i[A_\nu,\L]=\C_\nu +i\d[A_\nu,I]$ and we know from the above that 
$$
\lim_{\nu\rightarrow 0}\lim_{\mu\rightarrow\infty}\SPn{\chi_\mu\psi}{\C_\nu\chi_\mu\psi} = \SPn{\psi}{N\psi}.
$$ 
Furthermore, as $[A_\nu,I]$ is a well-defined operator on $\dom(N^{1/2})$ (see Lemma \ref{Lem:CommutatorN}) and has the strong limit \eqref{m13} for $\nu\rightarrow 0$, relation \eqref{Thm:Virial} follows from \eqref{m12} by first taking $\mu\rightarrow\infty$ and then $\nu\rightarrow 0$.
\end{proof}

\subsection{Eigenvalues of $\L$}

\begin{prop}
\label{propeigenvalues}
1. Let $\d$ be arbitrary and suppose $\psi$ is a normalized eigenvector of $\L$ with eigenvalue $e$. Then
\begin{eqnarray}
\|\bar P_\Omega \psi\| &\le& 10 c_2 |\d|, \label{m14}\\
{\rm dist}\big(e,{\rm spec}(\L_S)\big) &\le&  \textstyle\frac{2}{\sqrt 3} |\d| (1-\|\bar P_\Omega\psi\|^2)^{-1/2}, \label{m14'}
\end{eqnarray}
where $c_2$ is given in Lemma \ref{LemOperators}.

2. Suppose that $\d$ is small such that $\textstyle\frac{2}{\sqrt 3} |\d| (1-\|\bar P_\Omega\psi\|^2)^{-1/2} < \varepsilon/2$,  where $\varepsilon$ is the distance between the nearest eigenvalues of $\L_S$. Then, by \eqref{m14'}, there is a unique $e_0\in{\rm spec}(\L_S)$ which is closest to $e$. Let $P_{e_0}$ be the eigenprojection associated to this $e_0$ and denote $\bar P_{e_0}=\one_S-P_{e_0}$. Then (writing $P_{e_0} P_\Omega$ for $ P_{e_0}\otimes P_\Omega$)
\begin{equation}
\label{m14''}
\|\bar P_{e_0} P_\Omega\psi\|\le 2|\d|\varepsilon^{-1}. 
\end{equation}
\end{prop}
{\em Remark.\ } In point 2., which $e_0$ is closest to $e$ may depend on $\Delta$, and we are not proving that $e$ is continuously varying in $\Delta$.

\begin{proof}
1. Note that $P_\Omega i[A,I]P_\Omega=0$ since $A\Omega = d\Gamma(-i\partial_u)\Omega=0$. The virial identity \eqref{Thm:Virial} implies
\begin{equation}
\label{m15}
0=\SPn{\bar P_\Omega\psi}{(N+\d i[A,I])\bar P_\Omega\psi} + 2{\rm Re}\,\SPn{\bar P_\Omega\psi}{\d i[A,I]P_\Omega\psi}.
\end{equation}
Using $\|\phi(h)N^{-1/2}\bar P_\Omega\|\le (1+1/\sqrt2)\|h\|$ and relation \eqref{m13} one obtains the bound 
$$
\|[A,I]N^{-1/2}\bar P_\Omega\|\le 8(1+1/\sqrt{2}) \|f'_\beta\|\, (1+\|f_\beta\|)\le 4 c_2.
$$ 
It follows that 
\begin{eqnarray*}
|\SPn{\bar P_\Omega\psi}{\d i[A,I]\bar P_\Omega\psi}| &\le& 4 |\d| c_2 \|\bar P_\Omega\psi\|\, \|N^{1/2}\bar P_\Omega\psi\|,\\
2{\rm Re}\,|\SPn{\bar P_\Omega\psi}{\d i[A,I]P_\Omega\psi}| &\le& 8 |\d| c_2\|P_\Omega\psi\|\,\|N^{1/2}\bar P_\Omega\psi\|.
\end{eqnarray*}
We combine the last two inequalities with \eqref{m15} to arrive at 
$$
0\ge (1-\alpha)\|N^{1/2}\bar P_\Omega\psi\|^2 -24\alpha^{-1} \d^2c^2_2,
$$ 
for any $\alpha>0$. The choice $\alpha=1/2$ gives \eqref{m14}.

Next we show \eqref{m14'}. For any eigenvalue $e_0$ of $\L_S$, set $Q_{e_0}:= \bar P_{e_0}P_\Omega$. Projecting $\L\psi=e\psi$, $\|\psi\|=1$, onto the range of $Q_{e_0}$ gives $Q_{e_0}\psi = -\d (\L_S-e)^{-1} Q_{e_0} I\psi$. (The result to be proven is clearly true if $e=e_0$ so we may assume $e\neq e_0$.) Therefore, for any eigenvalue $e_0$ of $\L_S$,
\begin{equation}
\|Q_{e_0}\psi\|\le \frac{|\d|}{{\rm dist}(e,{\rm spec}(\L_S)\backslash\{e_0\})}.
\label{m25}
\end{equation}
Since $\sum_{e_0\in{\rm spec}(\L_S)} Q_{e_0} = 3P_\Omega$ we have $3\|P_\Omega\psi\|^2 = \sum_{e_0\in{\rm spec}(\L_S)}\|Q_{e_0}\psi\|^2\le 4\|Q_{e_*}\psi\|^2$, where $e_*$ is an eigenvalue of $\L_S$ maximizing the norm $\|Q_{e_0}\psi\|$. Using the latter bound in \eqref{m25} gives
$$
{\rm dist}(e,{\rm spec}(\L_S)\backslash\{e_*\})  \le \frac{|\Delta|}{\|Q_{e_*}\psi\|} \le \frac{2|\Delta|}{\sqrt3 \sqrt{1-\|\bar P_\Omega\psi\|^2}}.
$$
Since ${\rm dist}(e,{\rm spec}(\L_S))  \le {\rm dist}(e,{\rm spec}(\L_S)\backslash\{e_*\})$, we have shown \eqref{m14'}.


2. We have ${\rm dist}(e,{\rm spec}(\L_S)\backslash\{e_0\})>\varepsilon/2$ and \eqref{m14''} follows from \eqref{m25}. This concludes the proof of Proposition \ref{propeigenvalues}. 
\end{proof}

Instability of eigenvalues of $\L_0$ under the perturbation $\d I$ can be shown provided a (``Fermi Golden Rule''-)condition of effective coupling is satisfied.
\begin{prop}
\label{lsoprop}
Let $\Pi_0$ be the rank-two spectral projection onto the kernel of $\L_0$ and set $\bar\Pi_0=\one-\Pi_0$. The operator
$$
\Lambda_0 \equiv  \Pi_0 I\bar\Pi_0 (\L_0-i 0_+)^{-1} I\Pi_0 \equiv \lim_{\eta\rightarrow 0_+} \Pi_0 I\bar\Pi_0 (\L_0-i \eta)^{-1} I\Pi_0
$$
exists and is anti self-adjoint (it equals $i$ times a self-adjoint operator). The eigenvalues are ${\rm spec}(\Lambda_0) =\{0,i \Delta^{-2} \tau^{-1}\}$, where $\tau^{-1}$ is given in \eqref{t2}. Moreover, $\Lambda_0\psi_{S,\beta}=0$ (see \eqref{psiokms}).
\end{prop}

{\em Proof of Proposition \ref{lsoprop}. } We identify $\Lambda_0$ with a $2\times 2$ matrix relative to the orthonormal basis $\{\varphi_{++}\otimes\Omega,\varphi_{--}\otimes\Omega\}$ of ${\rm Ran}\Pi_0$. Here, $\varphi_{++}=\varphi_+\otimes\varphi_+$ and $\sigma_z\varphi_\pm=\pm\varphi_\pm$. We caclulate explicitly
\begin{eqnarray}
4\Lambda_0\varphi_{++} &=& \ \ \, \varphi_{++} \left\langle W(2f_\beta) (\L_R-\varepsilon-i 0_+)^{-1} W(2f_\beta)^*\right\rangle\nonumber\\
&& +\varphi_{++}  \left\langle JW(2f_\beta) J(\L_\r+\varepsilon-i 0_+)^{-1} JW(2f_\beta)^*J\right\rangle\nonumber\\
&& -\varphi_{--} \left\langle W(2f_\beta)^* (\L_\r+\varepsilon-i 0_+)^{-1} JW(2f_\beta)^*J\right\rangle\nonumber\\
&& -\varphi_{--}  \left\langle JW(2f_\beta)^*J (\L_\r-\varepsilon-i 0_+)^{-1} W(2f_\beta)^*\right\rangle.\label{2.27}
\end{eqnarray}
Here, $\langle\ \cdot\ \rangle=\langle\Omega, \cdot\ \Omega \rangle$. Since $J\Omega=\Omega$, $J e^{-\beta \L_\r/2} W(h_\beta)\Omega=W(h_\beta)^*\Omega$ (by properties of the modular conjugation $J$ and the modular operator $e^{-\beta \L_\r/2}$) and since $J e^{-\beta \L_\r/2}=e^{\beta \L_\r/2}J$, we have $\langle W(g_\beta) J W(h_\beta)J\rangle = \langle W(g_\beta) e^{-\beta \L_\r/2} W(h_\beta)^*\rangle$.  
A term $\langle W(g_\beta) J W(h_\beta)J\rangle$ can thus be calculated as the holomorphic continuation of ${\mathbb R}\ni t\mapsto \langle W(g_\beta) e^{i t\L_\r} W(h_\beta)^*\rangle$ at $t=i\beta/2$. For real values of $t$, the latter average is easy to calculate using that (1) the exponential generates a Bogoliubov dynamics ($t\mapsto e^{i ut} h_\beta$), (2) the CCR \eqref{ccr} and (3) that the thermal average is given by \eqref{thav}. The result is
\begin{equation*}
\langle W(g_\beta) J W(h_\beta)J\rangle = e^{\frac14\left( \SPn{g}{e^{-\beta|k|/2}h} -\SPn{h}{e^{\beta|k|/2}g}\right)} \ e^{-\frac14 \left( \SPn{g}{cg}+ \SPn{h}{ch} -\SPn{h}{ce^{\beta|k|/2}g} -\SPn{g}{c e^{-\beta|k|/2}h}\right)},
\label{2.28}
\end{equation*}
where, for short, 
\begin{equation}
c=\coth(\beta|k|/2).
\label{for short}
\end{equation}
Using the representation 
$(\L_\r-\varepsilon-i 0_+)^{-1} = i \lim_{\eta\downarrow 0} \int_0^\infty e^{i t(\varepsilon+i\eta)} e^{-i t\L_\r} d t$,
we cast \eqref{2.27} in the form $\Lambda_0\varphi_{++} = x(\varepsilon)\varphi_{++} +z(\varepsilon) \varphi_{--}$, where
\begin{eqnarray}
x(\varepsilon) &=& {\textstyle\frac12} i{\rm Re}\,\int_0^\infty e^{i t\varepsilon}\, e^{-2 i \SPn{f}{\sin(|k| t)f}} \, e^{-2\SPn{f}{c(1-\cos(|k| t))f}}\, d t\nonumber\\
z(\varepsilon) &=& -{\textstyle\frac12}i\int_0^\infty \cos(\varepsilon t) e^{-2\SPn{f}{\{c-\frac{2\cos(|k| t)}{e^{\beta|k|/2}-e^{-\beta|k|/2}}\}f}}\, dt,
\end{eqnarray}
with $c$ given in \eqref{for short}. The symmetry
$\sigma_x\otimes\sigma_x \ \Lambda_0(\varepsilon,f)\ \sigma_x\otimes\sigma_x =\Lambda_0(-\varepsilon,-f)$ 
(where we display the dependence on $\varepsilon$ and $f$ explicitly)  implies immediately that 
$\Lambda_0\varphi_{--} = z(\varepsilon)\varphi_{++} + x(-\varepsilon)\varphi_{--}$.
(Note that $z(\varepsilon)=z(-\varepsilon)$.) Therefore, the level shift operator takes the matrix form
\begin{equation}
\Lambda_0 =\left(
\begin{array}{cc}
x(\varepsilon) & z(\varepsilon)\\
z(\varepsilon) & x(-\varepsilon)
\end{array}
\right).
\end{equation}
By a deformation of the path of integration, it is not hard to verify that $x(\varepsilon)=-e^{\beta\varepsilon/2}z(\varepsilon)$ (see also Appendix E of \cite{Leggett}). This implies that the Gibbs state $\psi_{S,\beta}$, \eqref{psiokms}, is in the kernel of $\Lambda_0$. The other eigenvalue of $\Lambda_0$ is hence its trace, 
$$
{\rm Tr}\Lambda_0 = x(\varepsilon)+x(-\varepsilon)= i \int_0^\infty\cos(\varepsilon t)\cos\left(2\SPn{f}{\sin(|k| t)f}\right)   e^{-2\SPn{f}{c(1-\cos(|k| t))f}} d t.
$$
Using the relation \eqref{2.7} shows that ${\rm Tr}\Lambda_0= i \Delta^{-2} \tau^{-1}$, see \eqref{t2}. This completes the proof of Proposition \ref{lsoprop}. \hfill $\qed$

\begin{thm}
\label{eigenvaluethm} Suppose $0<|\d| < \d_0$, for some constant $\d_0$ given in \eqref{deltanot}. Then $\L$ has no eigenvalues except for a simple one at the origin. Moreover, $\L\psi_{\rm KMS}=0$, where $\psi_{\rm KMS}$ is the coupled KMS state \eqref{2.12}.
\end{thm}

\begin{proof} 
Let $e$ be an eigenvalue of $\L$ with associated normalized eigenvector $\psi$, and define, for $\eta>0$,
\begin{equation*}
X_\eta = \eta\, ((\L_0-e)^2+\eta^2)^{-1}=\textrm{Im}\,(\L_0-e- i\eta)^{-1}.
\end{equation*}
We derive an upper bound and a lower bound for
$$
q_e(\psi)=\d^2 \SPn{P_\Omega\psi}{ I \bar P_\Omega\, X_\eta I P_\Omega\psi}.
$$

{\em Upper bound.\ }
Since $\Delta\bar P_\Omega IP_\Omega\psi=\bar P_\Omega (\L-e)P_\Omega\psi=-\bar P_\Omega (\L-e)\bar P_\Omega\psi$, we have
\begin{align}
q_e(\psi)
=-\d \SPn{\bar P_\Omega\psi}{ (\L_0-e) \bar P_\Omega\, X_\eta I P_\Omega\psi}-\d^2 \SPn{\bar P_\Omega\psi}{ I \bar P_\Omega\, X_\eta I P_\Omega\psi}.
\end{align}
The bounds $\|I\|\le 1$, $\|X_\eta^{1/2}\|\le \eta^{-1/2}$ and $\|X_\eta^{1/2}(\L_0-e)\|\le \eta^{1/2}$ then imply that
\begin{align}
q_e(\psi)\le \eta^{1/2}\,\|\bar P_\Omega\psi\|\,q_e(\psi)^{1/2}+|\d| \eta^{-1/2} \|\bar P_\Omega\psi\|\, q_e(\psi)^{1/2}.      
\end{align}
Dividing by $q_e(\psi)^{1/2}$ and squaring gives
\begin{equation}\label{m29a}
q_e(\psi)\le (\eta +2|\d|+ \eta^{-1}\d^2)\,\|\bar P_\Omega \psi\|^2.
\end{equation}

{\em The lower bound.\ } Let $e=0$. With $\Pi_0=P_0 P_\Omega$ (recall the notation $P_0$ from Proposition \ref{propeigenvalues}) we get the lower bound
\begin{align}
\label{m30}
q_0(\psi)&\ge \d^2\|\bar P_\Omega X^{1/2}_\eta I\Pi_0\psi\|^2  +\d^2\|\bar P_\Omega X^{1/2}_\eta I\bar P_0P_\Omega \psi\|^2 -2\Delta^2 \|\bar P_\Omega X^{1/2}_\eta I\Pi_0\psi\|\ \|\bar P_\Omega X^{1/2}_\eta I\bar P_0P_\Omega \psi\| 
\\ \nonumber
& \ge \textstyle \frac12 \d^2\|\bar P_\Omega X^{1/2}_\eta I\Pi_0\psi\|^2 - \d^2\|\bar P_\Omega X^{1/2}_\eta I\bar P_0P_\Omega \psi\|^2\\ \nonumber
& \ge \textstyle \frac12 \d^2\SPn{\Pi_0 \psi}{ I \bar P_\Omega\, X_\eta I \Pi_0\psi}- \eta^{-1} \d^2 \|\bar{P}_0 P_\Omega \psi\|^2.
\end{align}
We link the first term on the right side to the level shift operator $\Lambda_0$ given in Proposition \ref{lsoprop}. Recalling that $X_\eta={\rm Im}(\L_0-i\eta)^{-1}$ and $\bar P_\Omega = \bar\Pi_0 +\bar P_0 P_\Omega$ we see that
\begin{equation}
\label{m-1}
\| \Pi_0 I\bar P_\Omega X_\eta I\Pi_0 - {\rm Im}\, \Pi_0 I\bar \Pi_0 (\L_0-i\eta)^{-1}I\Pi_0\|\le  \eta \varepsilon^{-2},
\end{equation}
since $\|{\rm Im}\, \Pi_0 I\bar P_0 P_\Omega (\L_0-i\eta)^{-1} I\Pi_0\|\le\|{\rm Im}\, \bar P_0 (\L_S-i\eta)^{-1}\|\le \eta \varepsilon^{-2}$, where $\varepsilon$ is the gap in the spectrum of $\L_S$. It follows from \eqref{m-1} and the definition of $\Lambda_0$ given in Proposition \ref{lsoprop} that ${\rm Im}\,\Lambda_0=\lim_{\eta\rightarrow 0_+} \Pi_0 I\bar P_\Omega X_\eta I\Pi_0$. The convergence speed is estimated in Lemma \ref{lemmac2},  \eqref{c20''}. Namely,
$$
\big| \SPn{\Pi_0\psi}{I\bar P_\Omega X_\eta I\Pi_0\psi}- {\rm Im}\, \SPn{\psi}{\Lambda_0\psi}\big| \le c\eta^{1/3} \|(1+\bar A^2)^{1/2}\bar P_\Omega I\Pi_0\psi\|^2\equiv c_5\eta^{1/3},
$$
where $c_5$ does not depend on $\psi$ (which is normalized).
Combining the last bound with \eqref{m30} and with 
$$
{\rm Im}\,\Lambda_0 = \Delta^{-2}\tau^{-1}\Pi_0(1-|\psi_{S,\beta}\rangle\langle\psi_{S,\beta}|) = \d^{-2}\tau^{-1}(\Pi_0-|\psi_{0,\rm KMS}\rangle\langle\psi_{0,\rm KMS}|)
$$ 
(see Proposition \ref{lsoprop} and where $\psi_{0,\rm KMS}=\psi_{S,\beta}\otimes\Omega$ is the unperturbed KMS state, \eqref{psiokms}), we obtain
\begin{equation}
\label{m-2}
q_0(\psi) \ge \textstyle\frac12\tau^{-1} \big( \|\Pi_0\psi\|^2- |\SPn{\psi_{0,\rm KMS}}{\psi}|^2\big) -\frac12 c_5\Delta^2\eta^{1/3}-\Delta^2 \eta^{-1} \|\bar P_0 P_\Omega\psi\|^2.
\end{equation}
We further decompose $\|\Pi_0\psi\|^2 = \|\psi\|^2-\|\bar P_\Omega\psi\|^2-\|\bar P_0 P_\Omega\psi\|^2$. Since $\psi$ is normalized, we arrive at the lower bound
\begin{equation}
\label{flb}
q_0(\psi) \ge \textstyle \frac12\tau^{-1} (1 -\|\bar P_\Omega\psi\|^2-\|\bar P_0 P_\Omega\psi\|^2- |\SPn{\psi_{0,\rm KMS}}{\psi}|^2) -\frac12 c_5\Delta^2\eta^{1/3}-\Delta^2 \eta^{-1} \|\bar P_0 P_\Omega\psi\|^2.
\end{equation}

\indent {\em The contradiction.\ } $\tau^{-1}$ is proportional to $\Delta^2$, see \eqref{t2}. We write $\tau^{-1}=\Delta^2\tau^{-1}_0$, with $\tau_0<\infty$ independent of $\Delta$. Combining the bounds \eqref{m29a} and \eqref{flb} and dividing by $\Delta^2$ gives 
\begin{eqnarray}
\textstyle \frac12\tau^{-1}_0 &\le&\textstyle  \frac12 \tau^{-1}_0|\SPn{\psi_{0,\rm KMS}}{\psi}|^2+  \textstyle  [\frac12\Delta^2\tau_0^{-1} +\eta +2|\d|+\eta^{-1}\d^2]\ \d^{-2}\|\bar P_\Omega\psi\|^2
\nonumber\\
&& \textstyle +[\frac12 \d^2\tau_0^{-1} +\eta^{-1}\d^2]\ \d^{-2}\|\bar P_0 P_\Omega\psi\|^2 +\frac12c_5\eta^{1/3}.
\label{m-3}
\end{eqnarray}
Suppose that $\L\psi=0$, $\|\psi\|=1$ and $\psi\perp\psi_{\rm KMS}$, where $\psi_{\rm KMS}$ is given in \eqref{2.12}. Then
\begin{equation}
\label{kmsconst}
|\SPn{\psi_{0,\rm KMS}}{\psi}| = |\SPn{\psi_{0,\rm KMS}-\psi_{\rm KMS}}{\psi}|\le \|\psi_{\rm KMS}-\psi_{0,\rm KMS}\|\le  |\d| c_{\rm KMS}.
\end{equation}
Here, an upper bound on $c_{\rm KMS}$ is readily obtained by estimating the power series expansion in $\d$ which relates $\psi_{\rm KMS}$ and $\psi_{0,\rm KMS}$, see e.g. \cite{A1} ($c_{\rm KMS}$ is proportional to $\beta$, the inverse temperature). Choosing $\eta=|\d|^{3/2}$ and using the bound \eqref{kmsconst} together with \eqref{m14} and \eqref{m14''} in \eqref{m-3} gives
$$
\textstyle\frac12 \tau_0^{-1}\le \d^2 \ \textstyle\frac12 \tau_0^{-1} (c^2_{\rm KMS}+c^2_3+4/\varepsilon^2) +|\d|^{3/2}c^2_3 +2|\d|c^2_3 +|\d|^{1/2}(c^2_3+4/\varepsilon^2+\frac12c_5).
$$
The latter inequality is violated for $|\d|<\d_0$, where 
\begin{equation}
\label{deltanot}
\d_0 := \min\big\{1, \big[ c^2_{\rm KMS}+c^2_3 +4/\varepsilon^2 +2\tau_0(4c^2_3+4/\varepsilon^2 +c_5/2)\big]^{-2}\big\}.
\end{equation}
This shows that $\L$ has a simple kernel if $|\d|<\d_0$.

To complete the proof one can proceed in two ways. One can either adapt the above argument to show directly instability of all nonzero eigenvalues of $\L_0$ under the perturbation $\d I$. Or one can invoke a general result, saying that if $\L$ has a simple kernel, then it does not have any nonzero eigenvalues \cite{JPNote}.
\end{proof}

\section{Proofs: Absolutely continuous Spectrum of $\L$}
\label{acsect}

To show that ${\rm spec} (\L)=\mathbb R$, we can use the Weyl criterion (see e.g. \cite{HiSi}, Theorem 5.10): $s\in\mathbb R$ is in the spectrum of $\L$ if and only if there is a sequence $\{\psi_n\}_n$ of normalized vectors in the domain of $\L$, satisfying $\lim_{n\rightarrow\infty}\|\L\psi_n -s\psi_n\|=0$. An explicit choice of $\psi_n$, for any $s\in\mathbb R$, is $\psi_n \propto a^*(f_n)\psi_{\rm KMS}$. Here, $\psi_{\rm KMS}$ is given in \eqref{2.12} and $f_n(u,\Sigma) = \sqrt{n/8\pi}\, \one_{[s-1/n, s+1/n]}(u)$.

\medskip

We now show {\em absolute continuity}. The spectrum of $\L$ in an interval $(a,b)\subset\mathbb R$ is purely absolutely continuous provided that for each vector $\varphi$ in some dense set, there is a constant $C(\varphi)<\infty$ such that 
\begin{equation}
\label{c2}
\liminf_{\epsilon\downarrow 0} \sup_{x\in(a,b)}\SPn{\varphi}{{\rm Im}(\L -x-i\epsilon)^{-1}\varphi}\leq C(\varphi).
\end{equation}
See for instance Proposition 4.1 of \cite{CFKS}. In order to control the boundary values of the resolvent, we expand it using the {\em Feshbach} map in \eqref{c3} below. For an operator $X$ acting on $\mathcal H$ we denote by $\bar X=\bar P X\bar P\upharpoonright_{{\rm Ran}\bar P}$ its restriction to the range of $\bar P=\one-P$, where $P=\one_{{\mathbb C}^2\otimes{\mathbb C}^2}\otimes |\Omega\rangle\langle\Omega|$ and $\Omega$ is the vacuum of \eqref{GluedFock}. For $z\in\mathbb C$ with ${\rm Im}z>0$, we define 
\begin{equation}
\label{c1}
\F(z) = P\big( \L-z -\Delta^2 I\bar P(\bar \L-z)^{-1}\bar PI\big)P,
\end{equation}
which we view as an operator on the range of $P$. The resolvent and reduced resolvent are related by
\begin{eqnarray}
(\L-z)^{-1} &=& (\bar\L-z)^{-1} +\F(z)^{-1} + (\bar\L-z)^{-1}\bar P\L P\F(z)^{-1} P\L\bar P (\L-z)^{-1}\nonumber\\
&& - \F(z)^{-1} P\L \bar P (\bar\L-z)^{-1} - (\bar\L-z)^{-1} \bar P\L P\F(z)^{-1}.
\label{c3}
\end{eqnarray}
Here, $(\bar\L-z)^{-1}$ is interpreted as an operator on ${\rm Ran}\bar P$. We have $P(\L-z)^{-1}P = \F(z)^{-1}$.

We introduce the family of norms
\begin{equation}
\label{norms}
\|\varphi\|_\kappa=\|\bar{N}^{-1/2}(1+\bar{A}^2)^{\kappa/2}\varphi\|,\quad  \varphi\in\bar P\mathcal H,\ \kappa\ge 0.
\end{equation}

\begin{thm}
\label{thmc0} 
Let $\varphi,\psi\in\dom(\bar{A}^2)\cap \dom(\bar N^{1/2})$, where $A=d\Gamma(-i\partial_u)$. Then 
$$
\left| \partial_z \SPn{\varphi}{(\bar\L-z)^{-1}\psi}\right|\le C \big( \|\varphi\|_2+\|\bar N\varphi\|_1\big)\big(\|\psi\|_2 +\|\bar N\psi\|_1\big),
$$
where $C$ is independent of $z\in{\mathbb C}_+$ and $\varphi,\psi$.
\end{thm}

Theorem \ref{thmc0}, together with the fact that ${\rm Ran} IP\subset\dom(A^2)\cap\dom(N^{1/2})$, implies that $\F(z)$, \eqref{c1}, extends continuously to ${\rm Im}z\ge 0$ (as a function with values in the operators on ${\rm Ran}P$). We denote its value for $x\in\mathbb R$ by $\F(x)$.
\begin{thm}
\label{thmc2}
For any real $x_0\neq 0$ there exist $r(x_0)>0$ and $c(x_0)<\infty$ such that 
\begin{equation}
\label{c4}
\| \F(x)^{-1}\|\le c(x_0)\quad \mbox{for all $x$ such that \, $|x-x_0|<r(x_0)$}.
\end{equation}
\end{thm}

Theorems \ref{thmc0} and \ref{thmc2}, together with the expansion \eqref{c3}, show that \eqref{c2} is satisfied for all $(a,b)$ not containing the origin. This means that the absolutely continuous spectrum of $\L$ (which is a closed set) is $\mathbb R$ and that the singular continuous spectrum is empty.

\subsection{Proof of Theorem \ref{thmc0}}

Let $a$ and $b$ be expressions depending on the quantities $z\in\CC_+$, $\eta>0$, $\d\in\RR$. We use the notation 
\begin{equation}
\label{notation}
a \ple b
\end{equation}
to mean that there is a constant $c$ which does not depend on any of the above quantities, such that $a\le cb$. We introduce a regularization $\bar\L(\eta)$, $\eta>0$, defined as follows. The domain is $\dom(\bar\L(\eta))=\dom(\bar \L_0)\cap \dom(\bar N)$ and 
\begin{equation}
\label{c5}
\bar\L(\eta) = \bar\L_0-i\eta\bar N +\Delta \bar{I}(\eta),\qquad \mbox{with}\qquad \bar I(\eta)=(2\pi)^{-1/2}\int_{\mathbb R} \widehat f(s) \tau_{\eta s}(\bar I)d s,
\end{equation}
and where $\tau_t(X)=e^{it \bar A}Xe^{-it \bar A}$. Here, $f$ is a Schwartz function satisfying $f(0)=f'(0)=f''(0)=f'''(0)=1$. In terms of the Fourier transform, this means 
\begin{equation}
\label{c6}
(2\pi)^{-1/2}\int_{\mathbb R} (is)^k\widehat f(s) d s=1,\quad\mbox{ for } k=0,1,2,3.
\end{equation}
We show in the next result some properties of $\bar\L(\eta)$. In particular, $\bar\L(\eta)-z$ is invertible for $\eta>0$, ${\rm Im}z>-\eta/2$. We denote the resolvent by
$$
R_z(\eta):=(\bar\L(\eta)-z)^{-1}.
$$
\begin{lem}
\label{lemmac1} There is a $\Delta_0>0$, independent of $\eta>0$, such that for $|\Delta|<\Delta_0$:

{\rm \bf (1)}  $-2\eta\bar N\le {\rm Im} \bar\L(\eta)\le -\frac{\eta}{2}\bar N$. In particular, any $z$ with ${\rm Im}z>-\eta/2$  is in the resolvent set of $\bar \L(\eta)$. Moreover, for such $z$, $\|(\bar\L(\eta)-z)^{-1}\|\le (\eta/2+{\rm Im}z)^{-1}$.

{\rm \bf(2)} For $\eta>0$ and ${\rm Im}z>-\eta/2$, we have ${\rm Ran}R_z(\eta)\subset\big( \dom(\bar N)\cap \dom(\bar\L_0)\big)$ and $R_z(\eta)$ leaves $\dom(\bar A)$ invariant.

{\rm \bf(3)} For all $\psi\in\bar P\mathcal H$ and all $z\in\mathbb C_+$, we have $\lim_{\eta\rightarrow 0_+}R_z(\eta)\psi=(\bar \L-z)^{-1}\psi$.

{\rm \bf(4)} For all $\psi\in\bar P\mathcal H$ we have $\|\bar{N}^{1/2}R_z(\eta)\psi\|\le \sqrt{2}\eta^{-1/2}\,|\SPn{\psi}{R_z(\eta)\psi}|^{1/2}$ and $\|\bar{N}^{1/2}R_z(\eta)\psi\|\le 2\eta^{-1}\|\bar{N}^{-1/2}\psi\|$. The same estimates hold for $R_z(\eta)$ replaced by $R_z(\eta)^*$.
\end{lem}

{\em Proof of Lemma \ref{lemmac1}. } (1) Using \eqref{c6} with $k=0$ we can write
\begin{equation}
\label{c8}
{\rm Im}\bar\L(\eta) = -\eta\bar N^{1/2}\Big( \one +\frac{\Delta}{\sqrt{2\pi}\,\eta}{\rm Im}\int_{\mathbb R}\widehat f(s) \bar N^{-1/2}(\tau_{\eta s}(\bar I)-\bar I)\bar N^{-1/2} d s\Big)\bar N^{1/2}.
\end{equation}
Now $\tau_{\eta s}(\bar I)-\bar I=\int_0^{\eta s}\partial_{s'}\tau_{s'}(\bar I) d s' = i\int_0^{\eta s} \tau_{s'}([\bar A,\bar I]) d s'$. Thus we have $\|\bar N^{-1/2}(\tau_{\eta s}(\bar I)-\bar I)\bar N^{-1/2}\|\le \eta|s|\,\|\bar N^{-1/2}[A,I]\bar N^{-1/2}\|$. The expression \eqref{m13} shows that the latter norm is bounded above by a constant $C'$. Therefore, \eqref{c8} implies $-\eta(1+C|\d|)\bar N\le {\rm Im}\bar\L(\eta)\le-\eta (1-C|\Delta|)\bar N$ for $\d$ small and where $C=C'(2\pi)^{-1/2}\int_{\mathbb R}|s \widehat f(s)|ds$. This gives the bound on ${\rm Im}\bar\L(\eta)$. Now
\begin{equation}
\label{f1}
\|\psi\|\,\|(\bar \L(\eta)-z)\psi\| \ge |\SPn{\psi}{(\bar\L(\eta)-z)\psi}|\ge  |{\rm Im} \SPn{\psi}{(\bar\L(\eta)-z)\psi}|\ge (\eta/2+{\rm Im}z)\|\psi\|^2.
\end{equation}
In the same way $\|(\bar \L(\eta)-z)^*\psi\|\ge (\eta/2+{\rm Im}z)\|\psi\|$. For ${\rm Im}z>-\eta/2$, $(\bar \L(\eta)-z)^*$ has trivial kernel and so ${\rm Ran}(\bar \L(\eta)-z)$ is dense. However, due to \eqref{f1} and since $\bar\L(\eta)-z$ is a closed operator, ${\rm Ran}(\bar \L(\eta)-z)$ is also closed, so it is all of $\bar P\mathcal H$. Therefore, the inverse of $\bar \L(\eta)-z$ is defined on the whole space and by the closed graph theorem it is bounded. The bound is obtained from \eqref{f1}. This shows (1).

To prove the first part of (2), note that $(\bar \L_0-i\eta\bar N-i) R_z(\eta) = \one +(-\Delta I(\eta)+z-i)R_z(\eta)$ is bounded. Hence $\bar \L_0R_z(\eta)=\bar \L_0 (\bar \L_0-i\eta\bar N-i)^{-1} (\bar \L_0-i\eta\bar N-i) R_z(\eta)$ is bounded as well. In the same way, $\bar N R_z(\eta)$ is bounded. It remains to show that $R_z(\eta)$ leaves $\dom(\bar A)$ invariant. For this, it suffices to prove that the derivative $\partial_t|_{t=0}$ of
$$
e^{i tA} R_z(\eta)\psi = \Big( \bar \L_0+(t-i\eta)\bar N +\Delta \tau_{t}(\bar I(\eta))-z\Big)^{-1} e^{it A}\psi
$$
exists, for any $\psi\in\dom(\bar A)$. One only needs to check that the derivative of the resolvent, at $t=0$, is bounded. This can be done easily by writing the derivative as the limit of the difference quotient and using the second resolvent equation for the numerator of the quotient. (2) follows.

(3) It suffices to show the result for any single, fixed $z_0$ in the upper half plane, e.g. $z_0=i$. This fact is seen by proceeding as in the proof of Theorem VIII.19 of \cite{RS1}, by expanding the resolvents in a power series in $z-z_0$. (Note that in the above reference, only self-adjoint operators are considered, but all that counts in the argument is the bound on the resolvent which we have established in point (1) of the present lemma). Let us show the result for $z=i$ now. First we note that $(\bar\L-i)^{-1}$ leaves $\dom(\bar N)$ invariant. A proof of this is obtained by expanding $(\bar\L-i)^{-1}$ into its Neumann series (in powers of $\Delta$) and using that $\bar N$ commutes with $(\bar\L_0-i)^{-1}$ and $\bar N^{-1} \bar I \bar N$ is bounded, so that $\bar N^{-1} (\bar\L-i)^{-1}\bar N$ is bounded. Therefore, for $\psi\in\dom(\bar N)$,
\begin{equation}
\label{c-1}
\big( (\bar\L(\eta)-i)^{-1}  - (\bar \L-i)^{-1}\big)\psi = (\bar\L(\eta)-i)^{-1}\big(i\eta\bar N -\Delta \bar I(\eta)+\Delta \bar I\big)(\bar \L-i)^{-1}\psi\rightarrow 0,
\end{equation}
as $\eta\rightarrow 0_+$. Finally, since $\dom(\bar N)$ is dense in $\bar P\mathcal H$ and $\|(\bar\L(\eta)-i)^{-1}  - (\bar \L-i)^{-1}\|\le 2$, \eqref{c-1} is valid for all $\psi\in\bar P\mathcal H$. This proves (3). 

(4) Due to (a), we have $\bar N\le -2\eta^{-1}{\rm Im}\bar\L(\eta)$, so
\begin{align*}
\|\bar{N}^{1/2}R_z(\eta)\psi\|^2&\le 2 \eta^{-1}\,|\SPn{R_z(\eta)\psi}{\Im(\bar{\L}(\eta)-z)R_z(\eta)\psi}|\\ 
\nonumber
&\le 2 \eta^{-1}\,|\SPn{\psi}{R_z(\eta)\psi}|\le  2 \eta^{-1}\,\|\bar{N}^{1/2}R_z(\eta)\psi\|\,\|\bar{N}^{-1/2}\psi\|.
\end{align*}
The estimate for $R_z(\eta)$ replaced by $R_z(\eta)^*$ is obtained in the same way. This shows (4) and concludes the proof of Lemma \ref{lemmac1}.\hfill\qed

\medskip

The operator
\begin{equation}
\label{keta}
K(\eta) :=[\bar{A},\bar{I}(\eta)]-\partial_\eta \bar{I}(\eta) 
= (2\pi)^{-1/2}\int_{\mathbb R} (1-is)\widehat f(s) \tau_{\eta s}([\bar A,\bar I]) d s,
\end{equation}
defined on $\dom(\bar N^{1/2})$, has the following properties.
\begin{lem}
\label{c-Lem-MK} 
Let $\varphi,\psi\in \bar P{\mathcal H}$.

(a) Assume $\||\partial_u|^{\alpha}f_\beta \|<\infty$ for some $1\le \alpha \le 2$. Then
\begin{equation}\label{c-MK6}
|\SPn{\varphi}{ K(\eta)\,\psi}|\le c\, \eta^{\alpha-1} \,||\bar{N}^{1/2}\varphi\|\,||\bar{N}^{1/2}\psi\|.
\end{equation}

(b) Assume $\||\partial_u|^{\alpha}f_\beta \|<\infty$ for some $2\le \alpha\le 3$. Then 
\begin{equation}\label{c-MK6a}
|\SPn{\varphi}{ K(\eta)\,\psi}|\le c\, \eta^3 \,||\bar{N}^{3/2}\varphi\|\,||\bar{N}^{1/2}\psi\|+c\,\eta^{\alpha-1} \,||\bar{N}^{1/2}\varphi\|\,||\bar{N}^{1/2}\psi\|.
\end{equation}
The constant $c$ does not depend on $\eta$. Both (a) and (b) hold if $K(\eta)$ is replaced by $K(\eta)^*$.
\end{lem}

%
%
%
%
%

Denote
$$
G_{\varphi,\psi,z}(\eta):= \SPn{\varphi}{(\bar \L(\eta)-z)^{-1}\psi}=\SPn{\varphi}{R_z(\eta)\psi}.
$$
\begin{lem}
\label{lemmac2}
Assume $\|\,|\partial_u|^\alpha f_\beta\|<\infty$ for some $\alpha>1$. There is a constant $c$ independent of $z\in{\mathbb C}_+$, $\eta>0$ and $\d$ with $|\d|<\d_0$, such that, for any $\varphi,\psi\in\dom(\bar A)$,
\begin{eqnarray}
\label{c20}
\left| G_{\varphi,\psi,z}(\eta)\right|&\le& c\, \|\varphi\|_1\,\|\psi\|_1\\
\label{c20'}
\|\bar{N}^{1/2} R_z(\eta)\varphi\|, \, \|\bar{N}^{1/2} R_z(\eta)^*\varphi\|&\le &c\,\eta^{-1/2}\, \|\varphi\|_1.
\end{eqnarray}
{}For any $x\in \RR$, $\SPn{\varphi}{(\bar{\L}-x-iy)^{-1}\psi}$ has a limit as $y\rightarrow 0_+$, denoted by $\SPn{\varphi}{(\bar{\L}-x-i0_+)^{-1}\psi}$,
and
\begin{equation}\label{c20''}
|\SPn{\varphi}{(\bar{\L}-x-i0_+)^{-1}\psi}-\SPn{\varphi}{(\bar{\L}-x-iy)^{-1}\psi}|\le c\, y^{\gamma/(1+\gamma)} \|\psi\|_1 \, \|\varphi\|_1,
\end{equation}
uniformly in $x\in\RR$, $y\in (0,1)$ and with $\gamma=\min\{1/2,\alpha-1\}$.
\end{lem}
\begin{proof}[ Proof of Lemma \ref{lemmac2}.] 
Using Lemma \ref{lemmac1}(2) we obtain
\begin{eqnarray}
\partial_\eta G_{\varphi,\psi,z}(\eta) &=& -\SPn{\varphi}{R_z(\eta) \big[\partial_\eta\bar{\L}(\eta)\big] R_z(\eta)\psi}\nonumber\\
&=& \SPn{\varphi}{ [\bar{A},R_z(\eta)]\psi}+
\Delta\SPn{\varphi}{\,R_z(\eta)\,K(\eta)\,R_z(\eta)\psi},
\label{c8'}
\end{eqnarray}
where $K(\eta)$ is defined in \eqref{keta}. Using the estimate \eqref{c-MK6} in \eqref{c8'} yields (recall the notation \eqref{notation})
\begin{align}
|\partial_\eta G_{\varphi,\psi,z}(\eta)| \ple & \
\|\varphi\|_1\,\| \bar N^{1/2} R_z(\eta)\,\psi\|+\|\psi\|_1\,\| \bar N^{1/2} R_z(\eta)^*\,\varphi\|
 \nonumber \\&
+  \eta^{\alpha-1} |\d|\, \|\bar{N}^{1/2} R_z(\eta)^*\varphi\|\,\|\bar N^{1/2}\,R_z(\eta)\psi\|.
\label{c8''}
\end{align}
Using Lemma \ref{lemmac1}(1) in \eqref{c8'} gives
\begin{align}\label{c-MK9}
|\partial_\eta G_{\varphi,\psi,z}(\eta)| \ple &\ 
\eta^{-1/2} \big (\|\varphi\|_1^2+\|\psi\|_1^2\big)
+
\,(\eta^{-1/2}+\eta^{\alpha-2})\,\big(| G_{\varphi,\varphi,z}(\eta)|+|G_{\psi,\psi,z}(\eta)|\big).
\end{align}
By Lemma \ref{lemmac1}(4) we have $\|\bar N^{1/2}R_z(1)\bar N^{1/2}\|\le 2$ and hence $|G_{\varphi,\psi,z}(1)|\le \|\psi\|_0\,\| \phi\|_0$. Taking $\varphi=\psi$ in \eqref{c-MK9} gives a differential inequality for $G_{\varphi,\varphi,z}(\eta)$ which implies \eqref{c20} for $\varphi=\psi$ by the standard Gronwall estimate \cite{Hartman}. Combining \eqref{c20} for $\varphi=\psi$ with Lemma \ref{lemmac1}(4) gives \eqref{c20'}. We can now use \eqref{c20'} in \eqref{c8''} to obtain 
\begin{equation}
\label{xxx}
|\partial_\eta G_{\varphi,\psi,z}(\eta)| \ple (\eta^{-1/2}+\eta^{\alpha-2})\|\varphi\|_1\|\psi\|_1.
\end{equation}
Integrating gives \eqref{c20}.

We now prove \eqref{c20''}. Let $0<\mu\ll \eta$ and $z,w\in \CC_+$. By the triangle inequality,
\begin{equation}
|\SPn{\varphi}{(R_z(\mu)-R_w(\mu))\psi}|
\le |\SPn{\varphi}{(R_z(\eta)-R_w(\eta))\psi}|+\sum_{v=w,z} |G_{\varphi,\psi,v}(\mu)-G_{\varphi,\psi,v}(\eta)|.
\end{equation}
Using the resolvent identity and \eqref{c20'} gives $|\SPn{\varphi}{(R_z(\eta)-R_w(\eta))\psi}|\ple  |z-w| \eta^{-1}\|\psi\|_1\,\|\varphi\|_1$. Next, it follows from \eqref{xxx} and \eqref{c20} that 
$$
|G_{\varphi,\psi,v}(\mu)-G_{\varphi,\psi,v}(\eta)|\le \int_\mu^\eta \big| \partial_\xi G_{\varphi,\psi,v}(\xi)\big| d \xi \ \ple \ \eta^{\gamma}\|\psi\|_1\, \|\varphi\|_1.
$$
Therefore, 
\begin{equation}
|\SPn{\varphi}{(R_z(\mu)-R_w(\mu))\psi}|
\ple (\eta^{-1}|z-w|+\eta^{\gamma})\, \|\psi\|_1\, \|\varphi\|_1.
\end{equation}
Thanks to Lemma \ref{lemmac1}(3) we may send $\mu \to 0$ and choose $\eta=|w-z|^{1/(1+\gamma)}$ to obtain
\begin{equation}
|\SPn{\varphi}{(\bar \L - z)^{-1}\psi}-\SPn{\varphi}{(\bar \L - w)^{-1}\psi}|
\ple  \,|z-w|^{\gamma/(1+\gamma)}\, \|\psi\|_1 \, \|\varphi\|_1.
\end{equation}
This shows the existence of $\lim_{y\to 0+} \SPn{\varphi}{(\bar \L -x- iy)^{-1}\psi}$ and
proves \eqref{c20''}.
\end{proof}

\medskip

For $\varphi$, $\psi\in\dom(\bar A)$, $z\in\mathbb C_+$, $\eta>0$, we define
\begin{equation}
\label{c31}
H_{\varphi,\psi,z}(\eta):= \partial_z\SPn{\varphi}{R_z(\eta)\psi} = \SPn{\varphi}{R_z(\eta)^2\psi}.
\end{equation}
Due to Lemma \ref{lemmac1} and \eqref{keta},
\begin{equation}
\label{c33}
\partial_\eta H_{\varphi,\psi,z}(\eta) = 
\SPn{\varphi}{[\bar A,R_z(\eta)^2]\psi}+S_1+S_2,
\end{equation}
where $S_1= \SPn{\varphi}{R_z(\eta)K(\eta) R_z(\eta)^2\psi}$ and $S_2=\SPn{\varphi}{R_z(\eta)^2 K(\eta) R_z(\eta)\psi}$. Taking into account $\|\bar N^{1/2} R_z(\eta)\| \ple \eta^{-1}$ (see Lemma \ref{lemmac1}(4)), \eqref{norms} and \eqref{c20'} we have 
\begin{equation}
\label{mmm10}
|\SPn{\varphi}{[\bar A,R_z(\eta)^2]\psi}|\ple  \eta^{-3/2}\|\varphi\|_1\,\|\psi\|_1.
\end{equation}
Next, due to \eqref{c-MK6a}, 
$$
|S_1| \ple \eta^3\|\bar N^{3/2}R_z(\eta)^*\varphi\|\, \|\bar N^{1/2}R_z(\eta)^2\psi\|+\eta^{\alpha-1}\|\bar N^{1/2}R_z(\eta)^*\varphi\|\, \|\bar N^{1/2}R_z(\eta)^2\psi\|.
$$  
Since $\|\bar{N}^{1/2}R_z(\eta)\bar{N}^{1/2}\|\ple \eta^{-1}$ (Lemma \ref{lemmac1}(4)) and using \eqref{c20'}, we get
\begin{equation}
\label{s1}
|S_1|\ple  \eta^{3/2}\,\|\bar{N}^{3/2}R_z(\eta)^*\varphi\| \,\|\psi\|_1+ \eta^{\alpha-3}\,\|\varphi\|_1\,\|\psi\|_1.
\end{equation}
A similar upper bound is obtained for $|S_2|$. We show the following result below.
\begin{lem}
\label{lemmac5}
Let $l\geq 0$, $\eta>0$, $z\in\mathbb C_+$ and $\psi\in\dom(\bar A)\cap\dom(\bar N^{l/2-1})$. Then
\begin{equation}
\label{c43}
\|\bar N^{l/2} R_z(\eta)\psi\| \le  C \big( \|\psi\|_1+\eta^{1/2} \|\bar N^{l/2-1}\psi\|\big)
\left\{
\begin{array}{ll}
\eta^{-(l+1)/2}, & \mbox{$l$ even}\\
\eta^{-l/2}, & \mbox{$l$ odd}
\end{array}
\right.
\end{equation}
The same statement holds if $R_z(\eta)$ is replaced by $R_z(\eta)^*$.
\end{lem}
To shorten notation we set
\begin{equation}
\label{newnorm}
\|\psi\|'_1:=  \|\psi\|_1+ \|\bar N^{1/2}\psi\|.
\end{equation}
Combining \eqref{c43}, for $l=3$, with \eqref{s1} gives 
\begin{equation}
\label{finals}
|S_1|+|S_2| \ple  \eta^{\alpha-3}\,\|\varphi\|'_1\,\|\psi\|'_1.
\end{equation}
With \eqref{c33} and \eqref{mmm10} we obtain $|\partial_\eta H_{\varphi,\psi,z}(\eta)| \ple (\eta^{-3/2}+\eta^{\alpha-3}) \|\varphi\|'_1\,\|\psi\|'_1 \ple \eta^{-3/2} \|\varphi\|_1'\,\|\psi\|_1'$ (as $\alpha \ge 2$). We integrate this estimate and obtain
\begin{equation}\label{c-MK10}
|H_{\varphi,\psi,z}(\eta)| \ple  \eta^{-1/2} \|\varphi\|_1'\,\|\psi\|_1'.
\end{equation} 
Finally, we consider again \eqref{c33}, but this time we write $\SPn{\varphi}{[\bar A,R_z(\eta)^2]\psi}= H_{\bar A\varphi,\psi,z}(\eta)-H_{\varphi,\bar A\psi,z}(\eta)$. Then, due to \eqref{c-MK10},
\begin{equation}
\label{imp}
|\SPn{\varphi}{[\bar A,R_z(\eta)^2]\psi}|\ple \eta^{-1/2} \|(1+\bar A^2)^{1/2}\varphi\|'_1\,\|(1+\bar A^2)^{1/2}\psi\|'_1.
\end{equation} 
According to \eqref{norms} and \eqref{newnorm}, $\|(1+\bar A^2)^{1/2}\varphi\|_1' = \|\varphi\|_2+\|\bar N\varphi\|_1$. We use the improved bound \eqref{imp}, together with \eqref{finals}, in \eqref{c33} to obtain $|\partial_\eta H_{\varphi,\psi,z}(\eta)| \ple (\eta^{-1/2}+\eta^{\alpha-3})( \|\varphi\|_2+\|\bar N\varphi\|_1)(\|\psi\|_2 +\|\bar N\psi\|_1)$. Integration gives the H\"older continuity
$$
|H_{\varphi,\psi,z}(\eta)- H_{\varphi,\psi,z}(\eta')| \ple |\eta-\eta'|^{\min\{1/2,\alpha-2\}}\ ( \|\varphi\|_2+\|\bar N\varphi\|_1)(\|\psi\|_2 +\|\bar N\psi\|_1).
$$
It follows that $H_{\varphi,\psi,z}(\eta)$ extends continuously to $\eta=0$, the extension satisfying $|H_{\varphi,\psi,z}(0)|\le C( \|\varphi\|_2+\|\bar N\varphi\|_1)(\|\psi\|_2 +\|\bar N\psi\|_1)$, with $C$ independent of $\varphi$, $\psi$ and $z\in{\mathbb C}_+$. By Lemma \ref{lemmac1}(3), the extension is $H_{\varphi,\psi,z}(0)=\SPn{\varphi}{(\bar \L-z)^{-2}\psi}= \partial_z\SPn{\varphi}{(\bar \L-z)^{-1}\psi}$. This concludes the proof of Theorem \ref{thmc0}, modulo the proofs of Lemmas \ref{c-Lem-MK} and \ref{lemmac5}, which we give now.

\medskip

\begin{proof}[Proof of Lemma \ref{c-Lem-MK}]
  Due to the definition \eqref{keta} of $K(\eta)$ and the expression \eqref{2.9''} for $I$, it is enough to show the estimates (a) and (b) for $\tau_{\eta s}([\bar A,\bar I])$ in \eqref{keta} replaced by $\tau_{\eta s}([A, W(g)]) = W(e^{\eta s\partial_u}g)\big(\phi(e^{\eta s \partial_u}g')-\textstyle\frac i2\SPn{g}{g'}\big)$, 
where $\||\partial_u|^{\alpha}g \|<\infty$ (see also \eqref{m4}). Hence it suffices to show the bounds (a) and (b) for $\SPn{\varphi}{\wt{K}(\eta)\psi}$, where 
\begin{equation}
\label{c-MK7}
\wt{K}(\eta)=\int_{\mathbb R} \,
W(e^{\eta s \partial_u}g)\,\big(\phi(e^{\eta s \partial_u}g')-\textstyle\frac i2\SPn{g}{g'}\big) d\mu(s),
\end{equation}
with $d\mu(s)=(2\pi)^{-1/2}(1-is)\widehat{f}(s)ds$. By \eqref{c6}, $\wt{K}(0)=0$, so the value of the integral \eqref{c-MK7} stays the same if we replace the integrand by
$$
{\mathcal I} = W(e^{\eta s \partial_u}g)\,\big(\phi(e^{\eta s \partial_u}g')-\textstyle\frac i2\SPn{g}{g'}\big)- W(g)\,\big(\phi(g')-\textstyle\frac i2\SPn{g}{g'}\big).
$$

{\em Proof of (a).} We write 
\begin{equation}
\label{c-MK1}
{\mathcal I}=\big(W(e^{\eta s\partial_u}g)-W(g)\big)\,\big(\phi(e^{\eta s \partial_u}g')-\textstyle\frac i2\SPn{g}{g'}\big)
+W(g)\ \phi(e^{\eta s \partial_u}g'-g')
\end{equation}
and estimate
\begin{equation}\label{c-MK5}
\|(W(e^{\eta s\partial_u}g)^*-W(g)^*)\varphi\| \le 
 \int_0^{\eta s}\big\|\partial_t W(-e^{t\partial_u}g)\varphi \big\| dt \ple \eta |s| \|\bar{N}^{1/2}\varphi\|.
\end{equation}
The last bound is obtained from $\partial_t W(-e^{t\partial_u}g)= \imath [A,W(-e^{t\partial_u}g)]$ and an application of \eqref{m4} (with $D=-\imath \partial_u$). It follows that 
\begin{equation}
\label{mmm1}
\left|\SPn{\varphi}{\big(W(e^{\eta s\partial_u}g)-W(g)\big)\,\big(\phi(e^{\eta s \partial_u}g')-\textstyle\frac i2\SPn{g}{g'}\big)\psi}\right| \ple \eta |s| \|\bar N^{1/2}\varphi\|\, \|\bar N^{1/2}\psi\|.
\end{equation}
Next we consider the remaining term in \eqref{c-MK1}. By the spectral theorem, 
$$
\|\phi(e^{\eta s \partial_u}g'-g')\psi\|
\ple \| e^{\eta s \partial_u}g'-g'\|\, \|\bar{N}^{1/2}\psi\|\ple  (|s| \eta)^\gamma \sup_{r\not=0}\tfrac{| e^{ ir}-1|}{|r|^\gamma}\,\| |\partial_u|^\gamma g'\|\, \|\bar{N}^{1/2}\psi\|
$$
and thus for any $\gamma\in[0,1]$, if $\|\, |\partial_u|^{1+\gamma}g\|<\infty$, then 
\begin{equation}\label{c-MK4}
\|\phi(e^{\eta s \partial_u}g'-g')\psi\|\ple  \eta^\gamma \,|s|^\gamma\ \|\bar{N}^{1/2}\psi\|.
\end{equation}
It follows that $|\SPn{\varphi}{W(g)\phi(e^{\eta s\partial_u}g'-g')\psi}| \ple \eta^\gamma |s|^\gamma \|\varphi\|\, \|\bar N^{1/2}\psi\|$. Combining this last estimate with \eqref{mmm1} yields $|\SPn{\varphi}{\wt{K}(\eta)\psi}| \ple (\eta +\eta^\gamma)\|\bar N^{1/2}\varphi\|\, \|\bar N^{1/2}\psi\|$, which proves \eqref{c-MK6}.

\medskip
{\em Proof of (b).} We write ${\mathcal I} = T_1+T_2+T_3$, with
\begin{eqnarray}
T_1 &=& \big( W(e^{\eta s\partial_u}g)-W(g)\big)\big(\phi(e^{\eta s\partial_u}g') -\phi(g')\big),\label{t1} \\
T_2 &=& \big( W(e^{\eta s\partial_u}g)-W(g)\big)\big(\phi(g')-\textstyle\frac i2\SPn{g}{g'}\big),\label{t2'}\\
T_3 &=& W(g)\big(\phi(e^{\eta s\partial_u}g') -\phi(g')\big).\label{t4}
\end{eqnarray}
Using the bounds \eqref{c-MK5} and \eqref{c-MK4} with $\gamma=1$ gives $|\SPn{\varphi}{T_1\psi}| \ple \eta^2 s^2 \|\bar N ^{1/2}\varphi\|\,\|\bar N ^{1/2}\psi\|$, so
\begin{equation}
\label{t1bound}
\left|\int_{\mathbb R} \SPn{\varphi}{T_1\psi}d\mu(s)\right| \ple \eta^2 \|\bar N ^{1/2}\varphi\|\,\|\bar N ^{1/2}\psi\|.
\end{equation}
Next, since $W(e^{\eta s\partial_u}g)-W(g)=\int_0^{\eta s}\partial_t W(e^{t\partial_u}g)dt=\int_0^{\eta s}W(e^{t\partial_u}g)\{i\phi(e^{t\partial_u}g')+\frac12\SPn{g}{g'}\}dt$, we obtain
\begin{equation}
\label{mmm3}
\int_{\mathbb R} \SPn{\varphi}{T_2\psi}d\mu(s) = \int_{\mathbb R} d\mu(s)  \int_0^{\eta s} dt\ \SPn{(T'_2+T''_2) \varphi}{(\phi(g')-\textstyle\frac i2\SPn{g}{g'})\psi}, 
\end{equation}
where
\begin{eqnarray}
T'_2 &=& W(-e^{t\partial_u}g)\{-i\phi(e^{t\partial_u}g') +  i\phi(g')\},\label{mmm5}\\
T''_2 &=& \big(W(-e^{t\partial_u}g) -W(-g)\big)\big( -i\phi(g')+\textstyle\frac 12\SPn{g}{g'}\big).
\label{mmm6}
\end{eqnarray}
Note that due to \eqref{c6}, $\int_{\mathbb R} d\mu(s)\int_0^{\eta s}dt \SPn{W(-g)(-i\phi(g')+\frac12\SPn{g}{g'})\varphi}{(\phi(g')-\textstyle\frac i2\SPn{g}{g'})\psi}=0$. In order to estimate the contribution to \eqref{mmm3} coming from \eqref{mmm5}, we use \eqref{c-MK4} with $\gamma=1$. The contribution of $T'_2$ to \eqref{mmm3} is $\ple \eta^2\|\bar N^{1/2}\varphi\|\, \|\bar N^{1/2}\psi\|$. Next we consider the contribution to \eqref{mmm3} coming from \eqref{mmm6}. The double integral in \eqref{mmm3} stays unchanged if we replace $T''_2$ by
$$
T'''_2=\int_0^t \big( \partial_r W(-e^{r\partial_u}g) - \partial_r|_{r=0} W(-e^{r\partial_u}g)\big)\big( -i\phi(g')+\textstyle\frac 12\SPn{g}{g'}\big) dr,
$$
since again, due to \eqref{c6}, the term containing $\partial_r|_{r=0} W(-e^{r\partial_u}g)$ vanishes. As $\partial_r W(-e^{r\partial_u}g) - \partial_r|_{r=0} W(-e^{r\partial_u}g)=\int_0^r \partial^2_x W(-e^{x\partial_u}g)dx$ and $\partial^2_x W(-e^{x\partial_u}g)=-\tau_x([A,[A,W(-g)]])$, which is an operator of the form $\tau_x(W(-g) P)$, where $P$ is a polynomial of degree two in field operators ($\phi(g')$ and $\phi(g'')$), we obtain $\|T_2'''\varphi\| \ple t^2 \|\bar N^{3/2}\varphi\|$. It follows that the contribution to \eqref{mmm3} coming from $T'_2$ is $\ple \eta^3\|\bar N^{3/2}\varphi\|\, \|\bar N^{1/2}\psi\|$. Hence
\begin{equation}
\label{mmm7}
\left|\int_{\mathbb R}\SPn{\varphi}{T_2\psi} d\mu(s)\right| \ple \eta^2\|\bar N^{1/2}\varphi\|\, \|\bar N^{1/2}\psi\|+ \eta^3\|\bar N^{3/2}\varphi\|\, \|\bar N^{1/2}\psi\|.
\end{equation}
Up to now, only two derivatives of $g$ are assumed to exist. Finally we estimate the term with $T_3$, 
\begin{equation}
\label{mmm8}
\int_{\mathbb R}d\mu(s) \int_0^{\eta s} dt \ \SPn{W(-g)\varphi}{\big( \partial_t\phi(e^{t \partial_u}g')- \partial_t|_{t=0}\phi(e^{t \partial_u}g')\big)\psi}
\end{equation}
where we inserted the term containing $\partial_t|_{t=0}\phi(e^{t \partial_u}g')$ for free, once again due to \eqref{c6}. 
Since $\partial_t\phi(e^{t \partial_u}g')- \partial_t|_{t=0}\phi(e^{t \partial_u}g') = \phi(e^{t\partial_u}g''-g'')$ we can apply the estimate \eqref{c-MK4} (with $g'$ replaced by $g''$ and $\eta s$ replaced by $t$) to obtain
\begin{equation}
\label{mmm9}
\left|\int_{\mathbb R}\SPn{\varphi}{T_3\psi} d\mu(s)\right| \ple \eta^{\gamma+1} \|\varphi\|\, \|\bar N^{1/2}\psi\|,
\end{equation}
for any $\gamma\in[0,1]$ and provided $\|\, |\partial_u|^{\gamma+2}g\|<\infty$. Combining \eqref{t1bound}, \eqref{mmm7} and \eqref{mmm9} yields the bound \eqref{c-MK6a}. The estimate for $K(\eta)^*$ is obtained in the same way.
\end{proof}

\bigskip

{\em Proof of Lemma \ref{lemmac5}. } For $l\ge 0$ we have 
\begin{equation}
\label{c44}
\bar N^{l/2} R_z(\eta)\psi = \bar N^{1/2} R_z(\eta)\bar N^{1/2}\bar N^{(l-2)/2}\psi +\bar N^{1/2} [\bar N^{(l-1)/2}, R_z(\eta)]\psi.
\end{equation}
The second term on the right side is
\begin{eqnarray}
\lefteqn{
\bar N^{1/2} [\bar N^{(l-1)/2}, R_z(\eta)]\psi = \Delta\, \bar N^{1/2} R_z(\eta) [I(\eta),\bar N^{(l-1)/2}] R_z(\eta)\psi}\nonumber\\
&=&\Delta\, \bar N^{1/2} R_z(\eta)\bar N^{1/2} \big( \bar N^{-1/2} I(\eta)\bar N^{1/2} - \bar N^{(l-2)/2} I(\eta) \bar N^{-(l-2)/2}\big) \bar N^{(l-2)/2} R_z(\eta)\psi.
\label{c45}
\end{eqnarray}
Using that $\|\bar N^{1/2} R_z(\eta)\bar N^{1/2}\|\le C\eta^{-1}$  (see Lemma \ref{lemmac1}(4)) and that, as we show below,
\begin{equation}
\label{c46}
\sup_{\eta>0} \|\bar N^\alpha\bar I(\eta)\bar N^{-\alpha}\| <\infty,
\end{equation}
for all $\alpha\in\mathbb R$, we obtain from \eqref{c44} and \eqref{c45} that
\begin{equation}
\label{c47}
\|\bar N^{l/2}R_z(\eta)\psi\| \ple  \eta^{-1}\|\bar N^{(l-2)/2}R_z(\eta)\psi\| + \eta^{-1}\|\bar N^{(l-2)/2}\psi\|.
\end{equation}
We now iterate \eqref{c47}. For $l$ even, we obtain after $l/2$ iterations
$$
\|\bar N^{l/2}R_z(\eta)\psi\| \ple  \eta^{-(l+1)/2}\|\psi\|_1 +\sum_{j=1}^{l/2} \eta^{-j} \|\bar N^{(l-2j)/2}\psi\| \ple \big( \|\psi\|_1+\eta^{1/2} \|\bar N^{(l-2)/2}\psi\|\big) \eta^{-(l+1)/2}.
$$
We use $\|\bar N^{1/2} R_z(\eta)\psi\|\le c\eta^{-1/2}\|\psi\|_1$ (see \eqref{c20'}) in the last iteration step. This gives \eqref{c43} for $l$ even. The estimate for $l$ odd is obtained in the same way, iterating \eqref{c47}.

It remains to show the bound \eqref{c46}, which is equivalent to $\|\bar N^\alpha W\bar N^{-\alpha}\|<\infty$, where $W= W(2if_\beta)$. Relation \eqref{m4} (with $D=\one$) gives, for any integer $m\ge 1$,
\begin{equation}
\bar N^m W\bar N^{-m} = \bar N^{m-1} W\big(\one +(\phi+c)\bar N^{-1}\big) \bar N^{-(m-1)} = \bar N^{m-1} W\bar N^{-(m-1)}B_m,
\label{bm}
\end{equation}
where $c$ is a constant, $\phi=\phi(2if_\beta)$ and $B_m= \bar N^{m-1}(\one +(\phi+c)\bar N^{-1}) \bar N^{-(m-1)}$. By using repeatedly the commutation relation $N\phi=\phi N +2^{-1/2}(a^*(2if_\beta)-a(2if_\beta))$ one sees that $B_m$ is bounded. Next, we show that $\bar N^{1/2} B_m\bar N^{-1/2}$ is bounded. It suffices to prove that $\bar N^{1/2}[B_m,\bar N^{-1/2}]$ is bounded. The representation $\bar N^{-1/2}=\pi^{-1}\int_0^\infty x^{-1/2} (\bar N+x)^{-1}d x$ (see \cite{K}, equation (3.53)) gives 
\begin{equation}
\label{bm3}
\bar N^{1/2}[B_m,\bar N^{-1/2}] = \pi^{-1}\int_0^\infty x^{-1/2} \bar N^{1/2} (\bar N+x)^{-1} [\bar N,B_m]  (\bar N+x)^{-1} dx.
\end{equation}
Next, $\|\bar N^{1/2} (\bar N+x)^{-1}\| \le (1+x)^{-1/2}$, $\|(\bar N+x)^{-1}\|\le (1+x)^{-1}$ and $[\bar N,B_m] = \bar N^{m-1} [\bar N,\phi]\bar N^{-m}$, which is easily seen to be bounded. The norm of the integrand in \eqref{bm3} is thus bounded above by a constant times $x^{-1/2}(1+x)^{-3/2}$, which is integrable in $x\in[0,\infty)$. Thus the operator \eqref{bm3} is bounded.

Iterating \eqref{bm} gives $\bar N^m W\bar N^{-m}= WB_1\cdots B_m$. Let $\alpha\ge 0$ and set $\alpha=m+\xi$, with $m=0,1,\ldots$ and $0\le \xi<1$. Then 
\begin{equation}
\label{bm1}
\bar N^\alpha W\bar N^{-\alpha} = \bar N^\xi \bar N^m W\bar N^{-m}\bar N^{-\xi}=\bar N^\xi WB_1\cdots B_m\bar N^{-\xi}.
\end{equation}
To show boundedness of $\bar N^\alpha W\bar N^{-\alpha}$ it suffices to show it for $\bar N^\xi[WB_1\cdots B_m,\bar N^{-\xi}]$, as $WB_1\cdots B_m$ is bounded. The representation $\bar N^{-\xi}=\pi^{-1}\sin(\pi\xi)\int_0^\infty x^{-\xi} (\bar N+x)^{-1}d x$ (see \cite{K}, equation (3.53)) gives 
\begin{equation}
\label{bm2}
\bar N^\xi [WB_1\cdots B_m,\bar N^{-\xi}] = \pi^{-1}\sin(\pi\xi)\int_0^\infty x^{-\xi}\bar N^\xi (\bar N+x)^{-1} [\bar N,WB_1\cdots B_m](\bar N+x)^{-1} dx.
\end{equation}
Using that $\|\bar N^\xi (\bar N+x)^{-1}\|\le (1+x)^{-1+\xi}$, $\|\bar N^{1/2}(\bar N+x)^{-1}\|\le (1+x)^{-1/2}$, and, as we show below, 
\begin{equation}
\label{bm4}
\|[\bar N,WB_1\cdots B_m]\bar N^{-1/2}\|<\infty,
\end{equation}
we see that the norm of the integrand in \eqref{bm2} is bounded from above by a constant times $x^{-\xi}(1+x)^{-3/2+\xi}$, which is integrable. To complete the proof of Lemma \ref{lemmac5}, we show \eqref{bm4}.

Expanding the commutator gives a sum of terms, each being of the form either $T_1=[\bar N,W]B_1\cdots B_m\bar N^{-1/2}$ or $T_2=WB_1\cdots B_k[\bar N,B_{k+1}]B_{k+2}\cdots B_m\bar N^{-1/2}$. $T_1$ is bounded since $[\bar N,W]\bar N^{-1/2}$ and $\bar N^{1/2} B_k\bar N^{-1/2}$ ($k=1,\ldots,m$) are. Finally, $T_2$ is bounded since $[\bar N,B_{k+1}]=\bar N^{k}[\bar N,\phi]\bar N^{-k-1}$ is. This shows \eqref{c46} (for $\alpha\ge 0$; for $\alpha\le 0$ the derivation is the same). The proof of Lemma \ref{lemmac5} is complete.\hfill \qed

\subsection{Proof of Theorem \ref{thmc2}}

We want to prove that $\F(x)$, $x\neq 0$, is invertible. For $y>0$,
\begin{equation}
\label{c50}
\F(x)-iy = \F(x+iy)\big\{ \one +\F(x+iy)^{-1} \big(\F(x)-iy -\F(x+iy)\big)\big\},
\end{equation}
where $\F(x+iy)^{-1}=P(\bar \L-x-iy)^{-1}P$. By Theorem \ref{thmc0}, $z\mapsto PI\bar P(\bar \L-z)^{-1}\bar PIP$ extends to a H\"older continuous map in $z\in\bar{\mathbb C}_+$, with exponent one. Hence
\begin{equation}
\label{c51}
\|\F(x)-iy- \F(x+iy) \| \le \Delta^2\|PI\bar P\big((\bar \L-x-iy)^{-1} - (\bar \L-x)^{-1} \big)\bar PIP\|\le C\Delta^2y,
\end{equation}
uniformly in $x\in\mathbb R$. Moreover, $x\neq 0$ is not an eigenvalue of $\L$, so $w-\lim_{y\rightarrow 0_+}iy(\L-x-iy)^{-1}=0$, from which it follows that 
\begin{equation}
\label{c52}
\lim_{y\rightarrow 0_+}iy\F(x+iy)^{-1}=0.
\end{equation}
Combining this with \eqref{c51} and \eqref{c50} shows for any $x\neq 0$ there is a $y_0$ s.t. if $|y|<y_0$, then $\F(x)-iy$ is invertible, and 
\begin{equation}
\label{c53}
\|(\F(x)-iy)^{-1}\| \le 2\|\F(x+iy)^{-1}\|\le (2 y)^{-1}.
\end{equation}
In the last step, we have again used \eqref{c52}. This implies that the kernel of $\F(x)$ is $\{0\}$. Indeed, if $\F(x)\psi=0$ for some $\psi\in{\rm Ran}P$, $\|\psi\|=1$, then $\|(\F(x)+iy)^{-1}\psi\| =1/y$. But \eqref{c53} gives $\|(\F(x)+iy)^{-1}\psi\|\le 1/(2y)$, a contradiction.

Since $\F(x)$ is invertible ($x\neq 0$) there is a constant $c_x$ s.t. $\|\F(x)^{-1}\|\le c_x$. We have $\F(x)^{-1} = \F(x')^{-1}[\one -(\F(x)-\F(x'))\F(x)^{-1}]$ and for $x'$ close enough to $x$, $\|(\F(x)-\F(x'))\F(x)^{-1}\|<1/2$. It follows that $\|\F(x')^{-1}\|\le 2c_x$. This completes the proof of Theorem \ref{thmc2}\hfill\qed

\bigskip
\bigskip

\noindent
{\bf Acknowledgement.\ } This work has been supported by an NSERC Discovery Grant and an NSERC Discovery Grant Accelerator.


\begin{thebibliography}{99}
%



\bibitem{ABG} W. Amrein, A. Boutet de Monvel, V. Georgescu: $C_0$-Groups, Commutator methods and Spectral Theory of $N$-body Hamiltonians, Progress in Mathematics Vol. 135, Birkh\"auser, Basel, 1996



\bibitem{A1}
H. Araki: {\em Relative Hamiltonian for Faithful Normal States
of a von Neumann Algebra}, Publ. RIMS, Kyoto Univ. {\bf 9}, 165-209  (1973)


\bibitem{AW}
 H. Araki, E. Woods: {\em Representations of the canonical commutation relations describing a nonrelativistic infinite free bose gas}, J. Math. Phys. {\bf 4}, 637-662 (1963)


\bibitem{BFS2}
V. Bach, J. Fr\"ohlich, I. M. Sigal: {\em Quantum Electrodynamics of Confined Nonrelativistic Particles}, Adv. Math. {\bf 137}, 299-395 (1998)


\bibitem{BFSrte} 
V. Bach, J. Fr\"ohlich, I.M. Sigal: {\em Return to equilibrium},  J. Math. Phys.  {\bf 41}, no. 6, 3985-4060 (2000)



\bibitem{BFSS}
V. Bach, J. Fr\"ohlich, I.M. Sigal, A. Soffer: {\em Positive commutators and the spectrum of Pauli-Fierz Hamiltonian of atoms and molecules}  Commun. Math. Phys. {\bf 207}, no. 3, 557-587 (1999)



\bibitem{BS}
A. Boutet de Monvel, J. Sahbani: {\em On the spectral properties of the Spin-Boson Hamiltonians}, Lett. Math. Phys. {\bf 44}, 23-33 (1998)



\bibitem{BR} O. Bratteli, D. Robinson, Operator Algebras and Quantum Statistical Mechanics 1,2, Texts and Monographs in Physics, Springer Verlag 1987


\bibitem{CFKS} H.L. Cycon, R.G. Froese, W. Kirsch, B. Simon: Schr\"odinger Operators, Texts and Monographs in Physics, Springer Verlag 1986


\bibitem{Davies}
E.B. Davies: Spectral Theory and Differential Operators. Cambridge studies in advanced mathematics 42, Cambridge University Press 1995


\bibitem{DJ}
J. Derezinski, V. Jaksic: {\em Spectral Theory of Pauli-Fierz Operators}, J. Funct. Analysis {\bf 180}, 243-327 (2001)


\bibitem{DJP} J. Derezinski, V. Jaksic, C.-A. Pillet: {\em Perturbation theory of $W^*$-dynamics, Liouvilleans and KMS-states}  Rev. Math. Phys. {\bf 15}, no. 5, 447-489 (2003)


\bibitem{FaMoSk}
J. Faupin, J.S. M\o ller, E. Skibsted: {\em Second Order Perturbation Theory for Embedded Eigenvalues}, Commun. Math. Phys. {\bf 306}, 193-228 (2011)




\bibitem{FM}
J. Fr\"ohlich, M. Merkli: {\em Thermal Ionization} Mathematical Physics, Analysis and Geometry {\bf 7}(3), 239-287 (2004)


\bibitem{FMRTE}
J. Fr\"ohlich, M. Merkli: {\em Another Return of ``Return to Equilibrium''} Commun. Math. Phys. {\bf 251}, 235-262 (2004)




\bibitem{Hartman}
P. Hartman: Ordinary differential equations, Wiley, New York, 1964 


\bibitem{HiSi}
P. Hislop, I.M. Sigal: Introduction to Spectral Theory, Applied Mathematical Sciences, Volume 113, Springer Verlag 1996



\bibitem{HuSp}
M. H\"ubner, H. Spohn: {\em Spectral properties of the spin-boson Hamiltonian} Ann. Inst. Henri Poincar\'e {\bf 62}, no.3, 289-323 (1995)


\bibitem{JP} V. Jaksic, C.-A. Pillet: {\em On a model for quantum friction. III. Ergodic properties of the spin-boson system},  Comm. Math. Phys.  {\bf 178}, no. 3, 627-651  (1996)


\bibitem{JPNote}
V. Jaksic, C.-A. Pillet: {\em A note on eigenvalues of Liouvilleans}, J. Stat. Phys. {\bf 218} 937-941 (2001)


\bibitem{K}
T. Kato: Perturbation Theory for Linear Operators. Die Grundlehren der Mathematischen Wissenschaften in Einzeldarstellungen, Volume 132, Springer Verlag  1966


\bibitem{Leggett} A.J. Leggett, S. Chakravarty, A.T. Dorsey, M.P.A. Fisher, A. Garg, W. Zwerger: {\em Dynamics of hte dissipative two-state system}, Rev. Mod. Phys. {\bf 59}, no.1, 1-85 (1987)


\bibitem{Merkli2000}
M. Merkli: {\em Positive Commutators in Non-Equilibrium Quantum Statistical Mechanics}, Commun. Math. Phys. {\bf 223}, 327-362 (2001)

\bibitem{MSB} M. Merkli,  I.M. Sigal, G.P. Berman: {\em Resonance Theory of Decoherence and Thermalization}, Annals of Physics {\bf 323}, 373-412 (2008)


\bibitem{Mourre}
E. Mourre: {\em Absence of Singular Continuous Spectrum for Certain Self-adjoint Operators}, Commun. Math. Phys. {\bf 78}, 391-408 (1981)


\bibitem{RS1}
B. Simon, M. Reed: Methods of Modern Mathematical Physics I, Functional Analysis. Academic Press, revised and enlarged edition 1980


\end{thebibliography}
\end{document}